\documentclass[10pt, letter, a4paper]{./lib/IEEEtran}
\usepackage{amsmath}
\usepackage{mathtools}
\usepackage{array}
\usepackage{cite}
\usepackage{algpseudocode,algorithm,algpseudocode}
\usepackage{float}
\usepackage{tabularx}
\usepackage{mathrsfs} 
\usepackage{tikz}
\usetikzlibrary{chains,patterns, arrows,shapes,positioning,arrows,decorations.markings,calc}
\usepackage{pgfplots}
\pgfplotsset{compat=newest}
\usepgfplotslibrary{groupplots}
\usepgfplotslibrary{fillbetween}
\usepackage{gensymb}
\usepackage{xcolor}
\usepackage[shortcuts,acronym]{glossaries}
\makeglossaries
\usepackage{fixmath}
\usepackage{amssymb}
\usepackage{bm}
\usepackage{bbm}
\usepackage{pbox}
\usepackage{multirow}
\usepackage{tabularx}
\pgfplotsset{compat=newest,
	/pgfplots/ybar legend/.style={
		/pgfplots/legend image code/.code={%
			\draw[##1,/tikz/.cd,bar width=3pt,yshift=-0.2em,bar shift=0pt]
			plot coordinates {(0cm,0.8em)};},
	},}

\newcommand\myeqb{\stackrel{\mathclap{\mbox{\scriptsize{($\eta=4$)}}}}{=}}

\usepackage{amsthm}
\theoremstyle{plain}
\newtheorem{theorem}{Theorem}
\newtheorem{lemma}{Lemma}
\newtheorem{approximation}{Approximation}
\newtheorem{proposition}{Proposition}
\newtheorem{remark}{Remark}

\usepackage{caption}

\usepackage{setspace}

\makeatletter
\def\therule{\makebox[\algorithmicindent][l]{\hspace*{.5em}\vrule height .75\baselineskip depth .25\baselineskip}}%

\newtoks\therules% Contains rules
\therules={}% Start with empty token list
\def\appendto#1#2{\expandafter#1\expandafter{\the#1#2}}% Append to token list
\def\gobblefirst#1{% Remove (first) from token list
	#1\expandafter\expandafter\expandafter{\expandafter\@gobble\the#1}}%
\def\LState{\State\unskip\the\therules}% New line-state
\def\pushindent{\appendto\therules\therule}%
\def\popindent{\gobblefirst\therules}%
\def\printindent{\unskip\the\therules}%
\def\printandpush{\printindent\pushindent}%
\def\popandprint{\popindent\printindent}%

\algdef{SE}[WHILE]{While}{EndWhile}[1]
{\printandpush\algorithmicwhile\ #1\ \algorithmicdo}
{\popandprint\algorithmicend\ \algorithmicwhile}%
\algdef{SE}[FOR]{For}{EndFor}[1]
{\printandpush\algorithmicfor\ #1\ \algorithmicdo}
{\popandprint\algorithmicend\ \algorithmicfor}%
\algdef{S}[FOR]{ForAll}[1]
{\printindent\algorithmicforall\ #1\ \algorithmicdo}%
\algdef{SE}[LOOP]{Loop}{EndLoop}
{\printandpush\algorithmicloop}
{\popandprint\algorithmicend\ \algorithmicloop}%
\algdef{SE}[REPEAT]{Repeat}{Until}
{\printandpush\algorithmicrepeat}[1]
{\popandprint\algorithmicuntil\ #1}%
\algdef{SE}[IF]{If}{EndIf}[1]
{\printandpush\algorithmicif\ #1\ \algorithmicthen}
{\popandprint\algorithmicend\ \algorithmicif}%
\algdef{C}[IF]{IF}{ElsIf}[1]
{\popandprint\pushindent\algorithmicelse\ \algorithmicif\ #1\ \algorithmicthen}%
\algdef{Ce}[ELSE]{IF}{Else}{EndIf}
{\popandprint\pushindent\algorithmicelse}%
\algdef{SE}[PROCEDURE]{Procedure}{EndProcedure}[2]
{\printandpush\algorithmicprocedure\ \textproc{#1}\ifthenelse{\equal{#2}{}}{}{(#2)}}%
{\popandprint\algorithmicend\ \algorithmicprocedure}%
\algdef{SE}[FUNCTION]{Function}{EndFunction}[2]
{\printandpush\algorithmicfunction\ \textproc{#1}\ifthenelse{\equal{#2}{}}{}{(#2)}}%
{\popandprint\algorithmicend\ \algorithmicfunction}%
\makeatother

\algrenewcommand\algorithmicprocedure{\textbf{Input}}
\algrenewcommand\algorithmicreturn{\textbf{Output:}}

\definecolor{Q1color}{RGB}{43, 140, 190}
\definecolor{Q2color}{RGB}{166, 189, 219}
\definecolor{Q3color}{RGB}{236, 231, 242}

\begin{document}
\bstctlcite{IEEEexample:BSTcontrol}
	
%
% paper title
% Titles are generally capitalized except for words such as a, an, and, as,
% at, but, by, for, in, nor, of, on, or, the, to and up, which are usually
% not capitalized unless they are the first or last word of the title.
% Linebreaks \\ can be used within to get better formatting as desired.
% Do not put math or special symbols in the title.
%\title{Advanced MEC-Aware Cell Association Strategies in 5G Mobile Systems}
%\title{MEC-based latency evaluations for V2X Cellular Communications}
\title{Prioritized Multi-stream Traffic in Uplink IoT Networks: Spatially Interacting Vacation Queues}
% author names and affiliations
% use a multiple column layout for up to three different
% affiliations

\iffalse
\author{Author 1,~\IEEEmembership{Student Member,~IEEE,}
	Author 2,~\IEEEmembership{Senior Member,~IEEE,}
	Author 3,~\IEEEmembership{Member,~IEEE}
	\thanks{Author 1 info.}% <-this % stops a space
	\thanks{Author 2 info.}% <-this % stops a space
	\thanks{Author 3 info.}}

\fi
\author{Mustafa~Emara,~\IEEEmembership{Student Member,~IEEE,}
	Hesham~ElSawy,~\IEEEmembership{Senior Member,~IEEE,}
	Gerhard~Bauch,~\IEEEmembership{Fellow,~IEEE}
	\thanks{M. Emara is with the Germany standards R\&D team, Next Generation and Standards, Intel Deutschland GmbH and  the Institute of Communications, Hamburg University of Technology, Hamburg, 21073 Germany (e-mail: mustafa.emara@intel.com)}% <-this % stops a space
	\thanks{H. ElSawy is with the Electrical Engineering Department, King Fahd University of Petroleum and Minerals, 31261 Dhahran, Saudi Arabia (email: hesham.elsawy@kfupm.edu.sa).}
	\thanks{G. Bauch is with the Institute of Communications,
		Hamburg University of Technology, Hamburg, 21073 Germany (email: bauch@tuhh.de).}}
\maketitle
\thispagestyle{empty}
% use for special paper notices\sqrt{}
%\IEEEspecialpapernotice{(Invited Paper)}
% make the title areag 
\maketitle
\thispagestyle{empty}
% As a general rule, do not put math, spxecial symbols or citations

%A
%B
\newacronym{BS}{BS}{base station}
%C
%D
\newacronym{DTMC}{DTMC}{discrete time Markov chain}
%E
\newacronym{EA}{EA}{equal allocation}
%F
\newacronym{5G}{5G}{fifth generation}
\newacronym{FCFS}{FCFS}{first come first serve}
%G
%H
%I
\newacronym{IoT}{IoT}{Internet of Things}
%J
%K
\newacronym{KPI}{KPI}{key performance indicator}
%L
\newacronym{LT}{LT}{Laplace transform}
\newacronym{LTE}{LTE}{long term evolution}
%M
\newacronym{MAM}{MAM}{matrix analytic method}
\newacronym{MTC}{MTC}{machine type communication}
\newacronym{MAC}{MAC}{medium access control}

%N
\newacronym{NB-IoT}{NB-IoT}{narrowband IoT}
%O
%p
\newacronym{PPP}{PPP}{Poisson point processes}
\newacronym{PDF}{PDF}{probability density function}
\newacronym{PMT}{PMT}{prioritized multi-stream traffic}
\newacronym{PA}{PA}{priority agnostic}
\newacronym{PAoI}{PAoI}{peak age of information}
%Q
\newacronym{QoS}{QoS}{quality of service}
\newacronym{QCI}{QCI}{QoS class identifier}
\newacronym{QBD}{QBD}{quasi-birth-death}
%R
\newacronym{RAT}{RAT}{radio access technology}
%S
\newacronym{SINR}{SINR}{signal to interference noise ratio}
\newacronym{SIR}{SIR}{signal to interference ratio}
%T
\newacronym{3GPP}{3GPP}{third generation partnership project}
\newacronym{TSN}{TSN}{time senstive networking}
\newacronym{TSP}{TSP}{transmission success probability}
%U
\newacronym{URLLC}{URLLC}{ultra reliable low latency communication}
%V
%w
\newacronym{WA}{WA}{weighted allocation}
%X
%Y
%Z
\thispagestyle{empty}
% !TEX root =../integration.tex
%%%%%%%%%%%%%%%%%%%%%%%%%%%%%%%%%%%%%%%%%%%%%%%%%%%%%%%%%%%%%%%%%%%%%%%%%%%%%%%%%%
% As a general rule, do not put math, special symbols or citations
% in the abstract or keywords.
\begin{abstract}

Massive Internet of Things (IoT) is foreseen to introduce plethora of applications for a fully connected world. Heterogeneous traffic is envisaged, where packets generated at each device should be differentiated and served according to their priority. This paper develops a novel priority-aware spatiotemporal mathematical model to characterize massive IoT networks with uplink prioritized multi-stream traffic (PMT). Stochastic geometry is utilized to account for the macroscopic network wide mutual interference between the coexisting devices. Discrete time Markov chains (DTMCs) are employed to track the microscopic evolution of packets within each priority queue. To provide a systematic and tractable model, we decompose the prioritized queueing model at each device to a single-queue system with server vacation. To this end, the IoT PMT network is modeled as spatially interacting vacation queues. Dedicated and shared channel priority-aware access strategies are presented. A priority-agnostic scheme is used as a benchmark to highlight the impact of prioritized uplink transmission on the performance of different priorities in terms of transmission probabilities and delay. Additional performance metrics as average number of packets, peak age of information, delay distribution, and Pareto frontiers for different parameters are presented, which give insights on stable operation of uplink IoT networks with PMT.

% Interactions between queues, in terms of the packet departure probabilities, occur due to mutual interference. Service vacations occur to lower priority packets to address higher priority packets.

\end{abstract}
%%%%%%%%%%%%%%%%%%%%%%%%%%%%%%%%%%%%%%%%%%%%%%%%%%%%%%%%%%%%%%%%%%%%%%%%%%%%%%%%%%
% Note that keywords are not normally used for peerreview papers.
\begin{IEEEkeywords}
Internet of Things, spatiotemporal models, priority queues, vacation queues,  queueing theory, stochastic geometry, grant-free access
\end{IEEEkeywords}
% no keywords
%%%%%%%%%%%%%%%%%%%%%%%%%%%%%%%%%%%%%%%%%%%%%%%%%%%%%%%%%%%%%%%%%%%%%%%%%%%%%%%%%%
% For peer review papers, you can put extra information on the cover
% page as needed:
% \ifCLASSOPTIONpeerreview
% \begin{center} \bfseries EDICS Category: 3-BBND \end{center}
% \fi
%
% For peerreview papers, this IEEEtran command inserts a page break and
% creates the second title. It will be ignored for other modes.
%\IEEEpeerreviewmaketitle
%%%%%%%%%%%%%%%%%%%%%%%%%%%%%%%%%%%%%%%%%%%%%%%%%%%%%%%%%%%%%%%%%%%%%%%%%%%%%%%%%%

\section{Introduction}\label{section:introduction}

The \ac{IoT} paradigm is paving the way to ensure connectivity, networking, and monitoring within different market segments \cite{Palattella2016}. Emerging segments entail, among other examples, smart cities, industrial \ac{IoT}, e-health, and cyber-physical systems, which are all tied with the \ac{IoT} technology advancement \cite{3GPP2018IoT}. Traffic prioritization schemes in IoT are inevitable due to the IoT heterogeneous traffic such as regular traffic (e.g., updates), query reposes (e.g., diagnostics), special measurements, control packets, warnings, and alarms \cite{Ayoub2018, elsawy2020spatial}. On the other hand, system alarms or failures need to be addressed almost immediately. Thus, heterogeneous multi-stream traffic is envisaged, where each traffic stream needs to be differentiated and addressed according to its priority. Such traffic discrepancies impose new challenges on how to properly model the network.

\subsection{Background and Motivation}

The necessity to meet the targeted \ac{QoS} becomes more prominent with \ac{PMT} in mixed-criticality systems. In such systems, the differentiated services and packets ought to be handled appropriately. For cellular systems, the concept of \ac{QCI} was first adopted in \ac{LTE} systems to characterize different services and to ensure that resources are allocated appropriately \cite{3GPP2019}. Each stream (i.e., bearer) has a corresponding \ac{QCI}, which indicates the service type, priority, and packet transmission requirements. Industrial automation is another sector that relies on \ac{PMT}, where guaranteed performance regarding successful packet delivery and latency is an imminent \ac{KPI} \cite{5GACIA2018}. In particular, the IEEE 802.1~Qbv amendment, among its many features, introduces eight different priority classes that are assigned to an incoming traffic stream which define the service requirements of each stream \cite{qbvieee2016}.

In addition to traffic prioritization within the network, massive number of deployed \ac{IoT} devices is foreseen \cite{Al-Fuqaha2015}. Due to the shared characteristic of the wireless channel, mutual interference between the \ac{IoT} devices is imminent. In this context, a key enabler of large scale \ac{IoT} devices is the low cost of deployment, which is realized via distributed and uncoordinated devices. Due to its decentralized nature, grant-free access is adopted in uplink cellular transmissions, where the scheduling complexities imposed by the scheduling grants from the \acp{BS} are alleviated \cite{Bader2017}. To this end, proper understanding and modeling of the prioritized traffic within the massive number of devices is required to i) characterize the performance , ii) understand the impact of different network parameters, iii) highlight common trends in the network’s performance, and iv) provide design insights.

\subsection{Related Work}

Queues with prioritized traffic have attracted wide attention in the queueing theory literature where different metrics (e.g. waiting time distribution and average queue length) are characterized \cite{S.Alfa2015,Takagi1991}. The incorporation of vacations to facilitate the analysis of priority queues is proposed in \cite{Doshi1986, Takagi1991, Machihara1996, HarcholBalter2005, Vuuren2007, Sleptchenko2015}. Nevertheless, the previously mentioned works consider only the interactions within a single queue and disregard the network-wide interaction between the devices \ac{PMT}. In addition, \ac{FCFS} based interactive queues are investigated in \cite{Rao1988, Luo1999} for the collision model which ignores the mutual interference between the devices. Interference characterization within large scale wireless networks has been facilitated in recent years via Stochastic geometry \cite{Andrews2011, ElSawy2013, Elsawy_tutorial}. However, an underlying limiting aspect of the stochastic geometry based models is the full buffer assumption, which assumes that the transmitter has always backlogged packets to be transmitted. Thus, conventional models based on stochastic geometry are oblivious to the temporal traffic evolution and the underlying queueing dynamics at each device.

To account for the temporal domain, recent efforts have integrated queueing theory with stochastic geometry, offering a full spatiotemporal characterization of the large-scale networks \cite{Zhong2017, Gharbieh2017,Gharbieh2018, Yang2019, Chisci2019, Chen2018, Dester2019}. In particular, the work in \cite{Zhong2017} characterizes the delay outage and downlink \ac{SIR} for a heterogeneous cellular network under random, \ac{FCFS} and round-robin scheduling schemes. The authors in \cite{Gharbieh2017} present a spatiotemporal characterization for grant-free uplink transmissions in IoT network, where the performance of power-ramping and back-off transmission strategies are investigated. The work in \cite{Gharbieh2017} is extended and compared to scheduled (i.e., grant-based) uplink transmissions in \cite{Gharbieh2018} and it is  shown that the network performance is highly dependent on the devices densities and traffic load. Analysis for small cell deployment is presented in \cite{Yang2019}, where the authors show the traffic load effect on the coverage probability. For an ad-hoc network, \cite{Chisci2019} presents a fine-grained spatiotemporal characterization for location-dependent \ac{QoS} classes in \ac{IoT} networks. 

Considering prioritized traffic under a spatiotemporal perspective, \cite{Chen2018} studies the delay and throughput in a cognitive radio setup, in which a network of secondary users share the channel with a single primary user. Secondary users are allowed to access the channel with a probability that depends on the primary user's queue length. However, their proposed framework only considers two priority classes. Recently, a framework to characterize an $N$-class prioritized devices is proposed in \cite{Dester2019}, where users randomly share the available channel. However, the model in \cite{Dester2019} is for prioritized devices (not traffic streams) and is only applicable to ad hoc networks. In summary, none of the aforementioned works consider PMT in uplink IoT networks. In addition, we are not aware of any work in the literature that characterizes the spatiotemporal performance, stability frontiers, and delay under different channels allocation strategies.
 
\subsection{Contributions}

When compared to the results presented in the aforementioned works, we provide an analytical framework that entails spatial macroscopic and microscopic scales of \ac{PMT} uplink large scale \ac{IoT} networks. The analysis relies on the joint utilization of stochastic geometry and queueing theory. The spatial macroscopic scale denotes the network-wide interactions arising between the devices in terms of the packet departure probabilities, due to mutual interference between the simultaneously active devices. Tools from stochastic geometry are employed to characterize the network-wide aggregate interference. On the other hand, the spatial microscopic scale, investigated via tools from queueing theory, represents the priority queues temporal dynamics and their interactions. To track the priority class being served at a given time stamp, a two-dimensional Geo/PH/1 \ac{DTMC} is employed for each device, where Geo stands for geometric inter-arrival process and PH stands for the Phase type departure process \cite{S.Alfa2015}. In summary, the main contributions of our are summarized as: 

\begin{itemize}	
	\item Develop a novel and tractable spatiotemporal framework, based on stochastic geometry and queueing theory, that jointly accounts for \ac{PMT} traffic in uplink large scale IoT networks;
	\item Employ a two dimensional  Geo/PH/1 \ac{DTMC} at every \ac{IoT} device  to account for the temporal evolution of queues in response to the \ac{PMT} arrivals and departures;
	\item Integrate the developed \acp{DTMC} within stochastic geometry framework to account for interference-based intrinsic inter-dependency between the macroscopic- and microscopic-scales;
	\item Compare the dedicated and shared allocation strategies with respect to various \acp{KPI}.
	%\item Present the Pareto frontiers that characterize the stability regions for different parameters. 
\end{itemize}

\subsection{Notation and Organization}
Throughout the paper, we adopt the following notation. Matrices and vectors are represented as upper-case and lower-case boldface letters ($\mathbf{A}$, $\mathbf{a}$), respectively. $[\mathbf{A}]_{i,j}$ denotes the $i$-th row and $j$-th column element of $\mathbf{A}$. The element wise Hadamard product is represented by the operator $\odot$. The function $\mathcal{Q}([\mathbf{A}]_{i,j},b)$ replaces the element in the $i$-th row and $j$-th column of $\mathbf{A}$ with the scalar $b$. The indicator function is denoted as $\mathbbm{1}_{\{a\}}$ which equals 1 if the expression $a$ is true and 0 otherwise. The $(\cdot)^{\text{T}}$ denotes the transpose operation. All ones and zeros vectors of dimension $m$ are represented as $\mathbf{1}_m$ and $\mathbf{0}_m$, respectively. In addition, $\mathbf{I}_m$ and $\mathbf{\mathcal{I}}_m$ denote, respectively, the identity and all ones matrices of dimension $m \times m$. The complement operator is denoted by the over-bar (i.e., $\bar{v} = 1-v$).  The notations $\mathbb{P}\{\cdot\}$ and $\mathbb{E}\{\cdot\}$ denote the probability of an event and its expectation. 

The rest of the paper is organized as follows. Section \ref{sec:system_model} provides the system model and the underlying physical and \ac{MAC} assumptions. The proposed queueing model along with the microscopic intra-device interactions among the priority queues are presented in Section \ref{sec:QT_anaylsis}. Section \ref{sec:sg_analysis} shows the macroscopic inter-device queueing interactions in terms  of mutual interference. Simulation results are presented in Section \ref{sec:simulation_results}. Finally, Section \ref{sec:Conclusion} summarizes the work and draws some conclusions.

\thispagestyle{empty}
% !TEX root =../Integration.tex
%%%%%%%%%%%%%%%%%%%%%%%%%%%%%%%%%%%%%%%%%%%%%%%%%%%%%%%%%%%%%%%%%%%%%%%%%%%%%%%%%%
\section{System Model}\label{sec:system_model}

\subsection{Spatial and Physical Layer Parameters}
This work studies a cellular uplink network, where the \acp{BS} and \ac{IoT} devices are spatially deployed in $\mathbb{R}^2$ according to two independent homogeneous \acp{PPP}, denoted by $\mathrm{\Psi}$ and $\mathrm{\Phi}$ with intensities $\lambda$ and $\mu$, respectively. Single antennas are employed at all devices and \acp{BS}. Grant-free access is assumed, where the devices attempt their transmissions on a randomly selected channel without a scheduling grant from their serving \ac{BS}. In addition, single connectivity is considered where each device is served by its nearest \ac{BS}. To alleviate congestion, a set of $C$ channels are utilized by the network and a priority-aware access strategy is adopted by the devices to access the available channels.\footnote{This corresponds to the Zadoff-Chu (ZC) codes utilized in LTE and 5G system for the random access channels to request scheduling grants\cite{Bader2017}. Data transmission follows a hybrid protocol with random access used for connection and re-connection, and with grant-based access for data transmission. For mathematical tractability, we consider only orthogonal channels (i.e., ZC codes stemming from the same root).} In this paper, we analyze three channel allocation strategies for priority-aware packet transmission, namely, i) dedicated strategy for each priority class with equal channel allocation, ii) dedicated strategy for each priority class with weighted channel allocation, and iii) shared strategy for all priority classes. For the dedicated strategy, each priority stream has an exclusive set of channels that can only be accessed by the devices to transmit their corresponding priority packets. For the shared strategy, all the channels can be accessed by all devices irrespective of the transmitted packet’s priority. 

An unbounded path-loss propagation model is adopted such that the signal power attenuates at the rate $r^{-\eta}$, where $r$ is the distance and $\eta > 2$ is the path-loss exponent. Small-scale fading is assumed to be multi-path Rayleigh fading, where the signal of interest and interference channel power gains $h$ and $g$, respectively, are exponentially distributed with unit power gain. All channel gains are assumed to be spatially and temporally independent and identical distributed (i.i.d.). Full path-loss channel-inversion power control is adopted, which implies that all devices adjust their transmit powers such that the received uplink average powers at their serving \ac{BS} is equal to a predetermined value $\rho$ \cite{ElSawy2014}. Moreover, a dense deployment of \acp{BS} is assumed, ensuring that every device is able to invert its path-loss almost surely. A packet generated at a given device is successfully decoded at its serving \ac{BS} if the received \ac{SINR} is larger than a predefined threshold $\theta$. In the case of successful decoding, the serving \ac{BS} transmits an ACK so the device can remove this intended packet from its respective queue. In the case of failed decoding, the serving \ac{BS} sends out an NACK and the packet remains at the head of the device's queue, awaiting a new transmission attempt in the next time slot. In this work, we assume error-free and negligible delay for ACK and NACK. Let $p_{i}$ and $\text{TSP}_i$ denote the coverage probability of an $i$-th priority packet and its \ac{TSP} given a transmission, respectively, which are evaluated as
\begin{align}\label{eq:SINR_1}
p_{i} &= \mathbb{P}\{\text{SINR}_i > \theta\}, \;\; \text{TSP}_i = \gamma_i p_{i},
\end{align}
where $\gamma_i$ is transmission probability of the $i$-th priority packet. 

\subsection{MAC Layer Parameters}

The proposed framework considers a synchronized, time slotted, and priority-aware system, in which packets of different priorities are generated at the devices. A \ac{PMT} model is considered such that packets are generated at each priority class independently of other classes. Hence, for a system with $N$ priority classes, batch arrivals up to size $N$ can occur in every time slot. Independent geometric inter-arrival times are assumed between packets belonging to each priority class with parameters $\alpha_i\in[0,1], \; i\in\{1,2,\cdots,N\}$. Through this work, traffic parameterized with lower indices has higher priority. 
\begin{figure}
	\begin{center}
		\ifCLASSOPTIONdraftcls
		\includegraphics[height=1.5in]{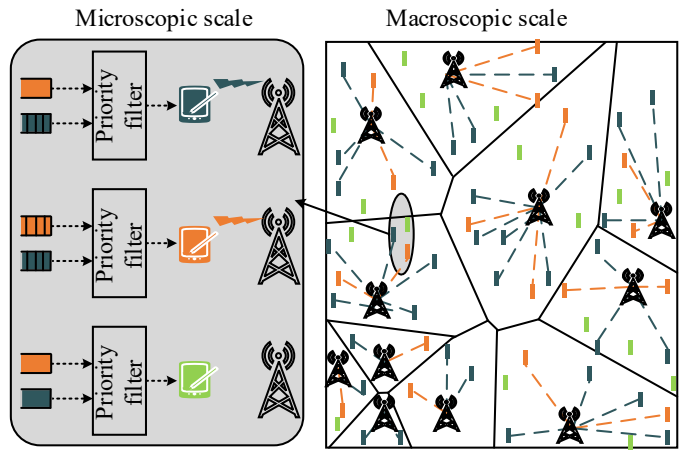}
		\vspace{-0.05in}
		\else
		\includegraphics[width=\columnwidth]{system_model/figures/dep_realization/spatiotemporal_realization.eps}
		\fi 		
		\caption{A snapshot realization of the network with two priority classes. Dark teal orange and green rectangles represent devices with first, second and no packets in their queues, respectively. The Voronoi cells of the \acp{BS} are denoted by the solid black lines while the dashed lines denote the active transmissions between devices and their serving \acp{BS}.}
		\label{fig:spatiotemporal_realization}
	\end{center}
\end{figure}
Generally, we consider that each device has $N$-priority finite queues, each of size $k_i$, that accumulate generated packets according to their priorities. The devices employ a priority-aware transmission strategy that prioritizes the transmission of high priority over lower priority packets. Furthermore, spatially-uniform distributed traffic is considered, whereas the case of location-dependent traffic can be extended by adopting different point processes (e.g. Poisson cluster point process) \cite{Haenggi2012}. It is assumed that arrival and departure of packets only occurs at the start of a time slot. If a high priority packet arrives while a lower priority queue is being addressed, service is interrupted and switched to the higher priority queue. The interrupted service is resumed after the high priority queue is empty. Thus, an inter-class preemptive discipline is considered along with an \ac{FCFS} discipline within each priority queue. In addition, \acp{BS} have no knowledge regarding the status of the devices queues. For the dedicated channel allocation strategies, the device randomly and uniformly selects one of the channels dedicated for the addressed packet priority. For the shared strategy, the device randomly and uniformly selects one of the complete set of channels regardless of the packet priority. In both cases the channel selection process is repeated in each transmission attempt. 

Pictorially, a snapshot realization of the network for two priority classes is shown in Fig.~\ref{fig:spatiotemporal_realization}. The right-hand side of the figure highlights a macroscopic network view and the left-hand side emphasizes the microscopic scale of three links. Due to the adopted preemptive priority discipline, imposed by the priority filter block, packets existing at high priority queues are prioritized for service (i.e., transmission) over packets existing in lower priority queues. If no high priority packets exist, the backlogged lower priority packets are served. In the case of having empty queues, no transmission is attempted and the device does not contribute to the network interference. It is worth noting that the time scale of channel fading, packet generation and transmission is much smaller than that of the spatial dynamics. Each spatial network realization for the adopted \acp{PPP} remains static over sufficiently large number of time slots, while channel fading, queue states, and device activities change from one time slot to another. 	
\thispagestyle{empty}
% !TEX root =../Integration.tex
%%%%%%%%%%%%%%%%%%%%%%%%%%%%%%%%%%%%%%%%%%%%%%%%%%%%%%%%%%%%%%%%%%%%%%%%%%%%%%%%%%
\section{microscopic queueing theory analysis}\label{sec:QT_anaylsis}

Throughout this section, a novel technique to model the \ac{PMT} is presented. In order to mathematically describe the different priority queues, a conventional way of characterizing the system is based on the following state space \cite[Chapter 9]{S.Alfa2015}. Let $\Delta = \Big\{(z_{1,n},z_{2,n}, \cdots,z_{N,n})| z_{i,n} \in \{0,1,\cdots,k_i\}; i \in\{1,2,\cdots,N\}\Big\}$ denote the state space, where $z_{i,n}$ denotes the number of $i$-th priority packets at the $n$-th time slot. Although tractable for the case of $N=2$, the depicted state space becomes disproportionately complex, when larger values of $N$ are considered \cite{S.Alfa2015}. Thus, a scalable and tractable model for a general number of priority classes is sought. To this end, the proposed vacation-based method to characterize the priority queues is more systematic and tractable for larger number of priority queues, when compared to conventional priority queues works \cite{Takagi1991, Machihara1996, Doshi1986, S.Alfa2015}. 

\subsection{Vacation Model: Motivation \& Description }
Priority queues can be modeled using vacation queues, where low priority queues are forced into a vacation period to allow the high priority queues service \cite{Takagi1991, Machihara1996, Vuuren2007, Sleptchenko2015, Doshi1986}. In other words, the \ac{PMT} is decomposed into a single queue with vacations, where the server becomes alternatively available and unavailable for a given priority class.\footnote{In our model, the server represents the wireless link over which a packet is transmitted. A sever vacation means that the IoT device is utilizing the current uplink time slot to transmit a high priority packet and no lower priority packet can be transmitted within this time slot.} The unavailability of the server, denoted as vacation, results from serving higher priority packets. An illustrative example for the vacation-based model is shown in Fig. \ref{fig:vacation_queues}. Due to its priority, the first priority queue is agnostic to the lower  priority queues dynamics. On the other hand, the second priority queue will be in vacation till the first priority queue is empty. Similarly, the third priority queue will be in vacation till the two higher queues are empty. Conceptually, a given queue will go strictly to vacation if a packet resides in any of the higher priority queues. 

For ease of demonstration, Fig. \ref{fig:vacation_queues}, assumes a hypothetical flawless server (i.e., $ p_i = 1; \forall i=\{1,2,3\}$), thus, ignoring the events of packet transmission failures due to poor wireless channel conditions or high aggregate network-wide interference from mutually active devices.
In that sense, one can consider that the vacation period of the $i$-th priority queue is the summation of the busy periods of the higher queues. In this context, the $i$-th priority queue's vacation period can be modeled via PH type distribution, which tracks the server's status whether it is serving the intended (i.e., $i$-th) priority queue or in vacation serving higher priority queues. By virtue of preemptive prioritization, there is no need to track any of the lower priority queues when analyzing the $i$-th priority class. 
\begin{figure}
	\begin{center}
		\includegraphics[width=\columnwidth]{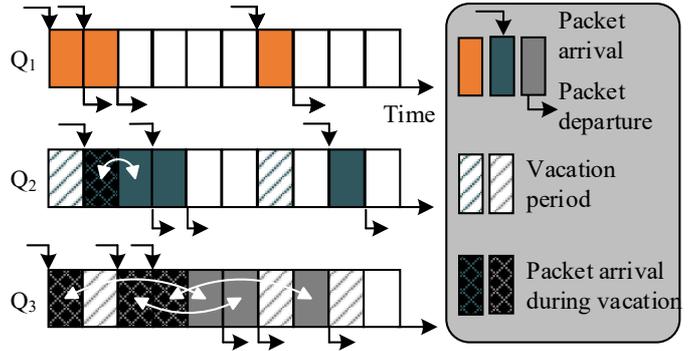}
		\caption{Vacation-based preemptive priority queues for $i=3$. White curved arrows indicates how low priority packets await service till higher priority queues are addressed.}
		\label{fig:vacation_queues}
	\end{center}
\end{figure}

Accordingly, the state space for the proposed vacation-based model $\Delta_{\text{v}}=\big\{\big(\mathcal{S}_i, \mathcal{V}_i \big)|i\in\{1,2,\cdots, N\}\big\}$, where $\mathcal{S}_i \in \{0,1,\cdots, k_i\}$ represents the number of packets at the i-th priority queue and $\mathcal{V}_i=\{(v_1,v_2,\cdots,v_{i-1}) | v_j \in \{0,1,\cdots, k_j\} \; \& \; \exists v_j >0 \} $ captures the vacation states of the server in terms of the number of packets in the higher priority queues. It is worth mentioning that, due to service preemption, any combination of non-empty higher priority queues is considered as a service vacation event for the lower priority queues.Utilizing such categorization of states, Fig. \ref{fig:DTMC_state} presents a two-dimensional Geo/PH/1 Markov chain that is employed at each \ac{IoT} device to track the packet's temporal evolution. The horizontal transitions represent the states of the server, denoted as phases, whether in vacation serving higher priority packets or serving the intended $i$-th priority queue. The vertical transitions represent the number of the packets in the $i$-th priority queue, denoted as levels. By virtue of the vacation-based categorization in $\Delta_{\text{v}}$, Fig. \ref{fig:DTMC_state} represents the transitions between serving the third priority class ($\mathcal{S}_3$ is captured via the left hand states) and being in vacation serving higher priority classes ($\mathcal{V}_3$ is captured via the right hand state and its internal components).

The PH type distribution of the server's vacation is represented via an absorbing Markov chain. In details, when serving higher priority packets, the server will be looping in the transient states of the PH type distribution. Absorbing Markov chains are mathematically described via an initialization vector and a transient matrix. In our case, the initialization vector and transient matrix are denoted as $\mathbf{v}_i$ and $\mathbf{V}_i$, respectively.  The initialization vector $\mathbf{v}_i$  captures all the possible initial states for vacations with their corresponding probabilities. That is, any combination of batch arrivals that include higher priority packets represents a legitimate initial state for the server vacation. The sub-stochastic transient matrix $\mathbf{V}_i$ tracks the server's vacation through tracking the temporal evolution of packets in the higher priority queues. Adopting this vacation-based model allows a systematic and tractable approach to model a network with generic $N$ priorities. %We will delve into the computation of the different matrices and vectors mentioned in Fig. \ref{fig:DTMC_state} which are required to characterize the queues evolution.

\begin{figure*} 
	\centering
	\input{QT_anaylsis/figures/DTMC_diagram/DTMC_diagram.tex}
	\caption{Two-dimensional DTMC modeling the vacation-based priority queues for $i=3$. States for the first, second and third priority classes are depicted by red, green and blue circles, respectively. Solid (dashed) lines are all multiplied by $\bar{\alpha}_3$ ($\alpha_3$).}
	\label{fig:DTMC_state}
\end{figure*}

\subsection{Vacation Model Analysis}
Let $m_i= \prod_{m=1}^{i-1} (k_m + 1)$ denote the number of transient states in the PH type distribution of the $i$-th priority queue and let the time-index $n$ be dropped hereafter. For mathematical convenience, we utilize a two level PH type distribution. In the higher level, absorption denotes packet departure from the $i$-th priority queue. At the lower level, absorption implies that the server comes back from vacation and is serving the $i$-th priority packet. Such hierarchy facilitates the transition matrix construction. The utilized higher level PH type distribution is denoted by the initialization vector and transient matrix tuple $(\bm{\beta}_i, \mathbf{S}_i)$, where $\bm{\beta}_i \in \mathbb{R}^{1 \times m_i}$ and  $\mathbf{S}_i \in \mathbb{R}^{m_i \times m_i}$.  In details, $\mathbf{S}_i$ is the sub-stochastic transient matrix that incorporates all the transition probabilities (including whether the server is in vacation or not) until packet departure~\cite{S.Alfa2015}. Starting from any state, the temporal evolution until a single packet departures is captured via the following absorbing Markov chain 
\begin{equation}
\mathbf{T}_i=\left[ \begin{array}{ll}{1} & \mathbf{0} \\ \mathbf{s}_i & \mathbf{S}_i \end{array}\right],
\end{equation}
where $\mathbf{s}_i\in \mathbb{R}^{m_i \times 1}$ is the probability of being absorbed from a given transient phase and is given by $\mathbf{s}_i = \mathbf{1}_{m_i}-\mathbf{S}_i\mathbf{1}_{m_i}$. It is worth noting that $\mathbf{s}_i$ only have a non-zero element in the location corresponding to the serving state of the server, since a packet only departs while the server is not in a vacation. 

Exploiting the mentioned PH type distribution, a scalable formulation that captures the queueing dynamics can be given in the form of a \ac{QBD} process~\cite{G.Kulkarni1999}. Let  $\mathbf{B}_{1,i}, \mathbf{C}_i$ and $\mathbf{B}_{2,i} \in \mathbb{R}^{m_i \times m_i}$ denote the boundary sub-stochastic matrices, \footnote{The stochastic transient boundary matrices, whose rows sum to one, capture the transitions between idle to idle, idle to serving $i$-th priority queue,  idle to vacation (serving $1\leq  j <i$ priority queues) and their complementary directions.} and let $\mathbf{A}_{0,i}, \mathbf{A}_{1,i}$ and $\mathbf{A}_{2,i} \in \mathbb{R}^{m_i \times m_i}$ represent the sub-stochastic matrices that capture the transition down a level, up a level, and in the same level within the \ac{QBD}, respectively. Accordingly, the probability transition matrix $\mathbf{P}_i$ of the $i$-th priority queue is
\begin{equation}\label{eq:QBD}
\mathbf{P}_i=\left[ \begin{array}{llllll} 
\mathbf{B}_{1,i} & \mathbf{C}_{i} &  &  &  &  \\
\mathbf{A}_{2,i} & \mathbf{A}_{1,i} & \mathbf{A}_{0,i} & &  &  \\
 & \mathbf{A}_{2,i} & \mathbf{A}_{1,i} & \mathbf{A}_{0,i} &  &  \\
 &  & \ddots  & \ddots & \ddots &  \\
 &  &  & & \mathbf{A}_{2,i} & \mathbf{B}_{2,i} 
\end{array}\right].
\end{equation}

In details, $\mathbf{B}_{1,i}=\bar{\alpha}_i\mathbf{S}_{0,i}$ captures all transitions from and to the idle state, where $\mathbf{S}_{0,i}$ is the stochastic transient boundary matrix. Similarly, $\mathbf{C}_{i}=\alpha_i\mathbf{S}_{0,i}$ captures the transitions to level 1, that represents an increment of the $i$-th priority packets. The forward transitions sub-matrix  $\mathbf{A}_{0,i}= \alpha_i\mathbf{S}_{i}$ represents the case where a new packet arrives and no packet departs (i.e., vacation state or serving state with transmission failure). The local transitions sub-matrix $\mathbf{A}_{1,i}={\alpha}_i\mathbf{s}_i\bm{\beta}_i + \bar{\alpha}_i\mathbf{S}_{i}$ represents no packet arrival while in transient state or a simultaneous arrival of one packet and a departure of another packet of the same priority. The backward transitions sub-matrix $\mathbf{A}_{2,i}=\bar{\alpha}_i\mathbf{s}_i\bm{\beta}_i$ captures the case of a packet being dispatched, leading to a decrement of the $i$-th queue packets. Finally, the boundary sub-matrix $\mathbf{B}_{2,i} = \alpha_i\mathbf{s}_i\bm{\beta}_i + \mathbf{S}_{i}$ captures the events when the $i$-th queue is full. Note that packets of the $i$-th priority that arrive in this state are lost due to queue overflow. Due to the embedded vacation model, the initialization vector is expressed as $\bm{\beta}_i = [1 \;\; \mathbf{0}_{m_i - 1}]$ has only $1$ at the serving state and zeros otherwise. 

In order to construct the \ac{QBD} via (\ref{eq:QBD}), the stochastic transient matrices $\mathbf{S}_{i}, \mathbf{S}_{0,i}$ are required. We first present preliminary definitions that facilitate the construction of $\mathbf{S}_{0,i}$ and $\mathbf{S}_{i}$. Let $\chi_i$ denote the probability that server starts a vacation while serving the $i$-th priority queue. Due to the adopted preemptive priority discipline, a vacation starts upon the arrival of any of the higher priority packets. Exploiting the independence between all arrival streams, $\chi_i$ is given by
\begin{align}\label{eq:chi}
	\chi_i 	&= 1- \prod_{m=1}^{i-1}\bar{\alpha}_m.
\end{align}

Let $\mathbf{v}_i \in \mathbb{R}^{1 \times m_i -1}$ denotes the vacation initialization vector, which have only non-zero values at the legitimate initial vacation states. The two level PH type distribution used to build the  \ac{QBD} in \eqref{eq:QBD} is constructed through the following proposition. 

\begin{proposition}\label{lem:S_So}
	The stochastic transient matrices of the $i+1$-th priority queue with coverage probability $p_i$, for the boundary $\mathbf{S}_{0,i+1}$ and non-boundary $\mathbf{S}_{i+1}$ states, are evaluated as
	\ifCLASSOPTIONdraftcls
	\begin{align}\label{eq:S_matrix}
	\mathbf{S}_{0,i+1} &= \tilde{\mathbf{S}}_{i+1}, \; \; \mathbf{S}_{i+1} = \tilde{\mathbf{S}}_{i+1}  \odot \mathcal{Q}([\mathbf{\mathcal{I}}_{m_i}]_{1,1},\bar{p}_i), \text{ such that }
	\tilde{\mathbf{S}}_{i+1}=\left[ \begin{array}{ll} 
	\bar{\chi}_i & \chi_i\mathbf{v}_i \\
	\tilde{\mathbf{v}}_{i}  & \mathbf{V}_{i}
	\end{array}\right], \nonumber
	\end{align}
	\else
	\begin{align}\label{eq:S_matrix}
	\mathbf{S}_{0,i+1} &= \tilde{\mathbf{S}}_{i+1}, \; \; \mathbf{S}_{i+1} = \tilde{\mathbf{S}}_{i+1}  \odot \mathcal{Q}([\mathbf{\mathcal{I}}_{m_i}]_{1,1},\bar{p}_i), \nonumber
	\end{align}
	such that $\tilde{\mathbf{S}}_{i+1}$ is defined as 	
	\begin{equation}
	\tilde{\mathbf{S}}_{i+1}=\left[ \begin{array}{ll} 
	\bar{\chi}_i & \chi_i\mathbf{v}_i \\
	\tilde{\mathbf{v}}_{i}  & \mathbf{V}_{i}
	\end{array}\right], \nonumber
	\end{equation}
	\fi 
	 where $\mathbf{V}_{i} = \mathbf{D}_{i}\mathbf{P}_{i}\mathbf{D}_{i}^{\text{T}}$ is the vacation visit matrix, $\chi_i$ is given in (\ref{eq:chi}), $\mathbf{v}_i$ is the vacation initialization vector and $\tilde{\mathbf{v}}_i = \mathbf{1}_{m_i-1}-\mathbf{V}_{i}\mathbf{1}_{m_i-1}$ is the absorption vector. In addition, $\mathbf{D}_{i-1} = [\bm{0}_{m_i-1} \; \; \mathbf{I}_{m_i-1}]$ is the selection matrix .
\end{proposition}
\begin{proof}
	See Appendix \ref{se:Appendix_A}.
\end{proof}

Based on Proposition 1, the vacation states are initialized through the vector $\chi_i\mathbf{v}_i$ (black arrows in Fig. \ref{fig:DTMC_state}), while all the vacation phases are captured by $\mathbf{V}_{i}$ (golden arrows in Fig. \ref{fig:DTMC_state}). Successful transmission of higher priority packets (i.e., end of vacation) is captured by $\tilde{\mathbf{v}}_{i}$ (green arrows in Fig. \ref{fig:DTMC_state}). At this point, the steady state distribution of each priority queue at each device can be evaluated. Let  $\mathbf{A}_i = \mathbf{A}_{0,i} + \mathbf{A}_{1,i} + \mathbf{A}_{2,i}$ and let $\bm{\pi}_i$ represent the unique solution of $\bm{\pi}_i\mathbf{A}_i = \bm{\pi}_i$, with the normalization condition $\bm{\pi}_i\mathbf{1}_{m_i} = 1$. Since finite queues are considered at the devices, one is interested to determine the critical arrival probability after which the probability of having full queues starts to dominate and the queues tend to be always non-empty  \cite{Loynes1962}. Through the rest of the paper, we use the term overflow (non-overflow) region to denote operating beyond (below) such a probability. Mathematically, for the \ac{DTMC} in (\ref{eq:QBD}) to be in the non-overflow region, the following condition must be satisfied
\begin{equation}\label{eq:stability}
	\bm{\pi}_i\mathbf{A}_{2,i} \mathbf{1}_{m_i} > \bm{\pi}_i\mathbf{A}_{0,i} \mathbf{1}_{m_i}.
\end{equation}
The condition in (\ref{eq:stability}) ensures that the departure probability of packets is higher than the arrival probability of packets, which ensures a low overflow probability. Consequently, the overflow probability can be highly reduced by increasing the queue size. %When (\ref{eq:stability}) is not satisfied, this implies that the packet departures cannot cope with the packet arrivals.

Let $\mathbf{x}_i = \big[\mathbf{x}_{i,0} \; \mathbf{x}_{i,1} \; \cdots \; \mathbf{x}_{i,k_i}\big]$ be the steady state probability vector where $\mathbf{x}_{i,j}$ incorporates the joint probabilities of having $j$ $i$-th priority packets and all possible combinations of number of packets with priority higher than $i$. In particular, let $\mathbb{P}\{n_1,n_2,n_3,\cdots,n_i\}$ 	denotes the joint probability of having $n_1$ packets at the first priority queue, $n_2$ packets at the second priority queue and so on until $n_i$ packets at the $i$-th priority queue, $\mathbf{x}_{i,j}$ can be represented as 
\ifCLASSOPTIONdraftcls
	\begin{align}\label{ref:JointProb}
	\mathbf{x}_{i,j} = \Big[\mathbb{P}\{(\underbrace{0,\cdots, 0}_\textrm{$i-1$}, j)\}\cdots \mathbb{P}\{(k_1,\cdots, 0, j)\} \cdots
	\mathbb{P}\{(k_1,\cdots, 1, j)\}
	\cdots  \mathbb{P}\{(k_1,\cdots, k_{i-1}, j)\} \Big]. \nonumber
	\end{align}
\else
	\begin{align}\label{ref:JointProb}
	\mathbf{x}_{i,j} &= \Big[\mathbb{P}\{(\underbrace{0,\cdots, 0}_\textrm{$i-1$}, j)\}  \cdots \mathbb{P}\{(k_1,\cdots, 0, j)\} \cdots \nonumber\\ 
	&\mathbb{P}\{(k_1,\cdots, 1, j)\} \cdots  \mathbb{P}\{(k_1,\cdots, k_{i-1}, j)\} \Big].
	\end{align}
\fi 		
In addition, let the scalar $x_{i,j}$ represents the probability of having $j$ packets in the $i$-th priority queue, which is evaluated as ${x}_{i,j} = \mathbf{x}_{i,j}\mathbf{1}_{m_i}$. By virtue of the adopted preemptive discipline, it is clear that the third priority queue is only granted service when all higher priority queues are empty. Thus, the transmission probability $\gamma_i$ can be computed as 
\begin{equation} \label{gggg}
	\gamma_i = \sum_{z_{i}=0}^{k_{i}} \mathbb{P}\{(0,0,\cdots, 0, z_i)\},
\end{equation}
whereas for the first priority queue $\gamma_1 = 1$. Let $r_i = m_i (k_i +1)$ be the number of possible states for the $i$-th queue, then the steady state solution for a stable system is characterized as follows.
\begin{lemma}\label{lem:SSP}
	The steady state distribution for the $i$-th queue with state transition matrix $\mathbf{P}_i$ is 
	\begin{equation}\label{eq:DirectSol_SSP}
	\mathbf{x}_i = \mathbf{1}_{r_i}\big(\mathbf{P}_i - \mathbf{I}_{r_i} + \mathcal{I}_{r_i} \big)^{-1}.
	\end{equation}
	
\end{lemma}
\begin{proof}
		Since we are considering finite \ac{DTMC} based on (\ref{eq:QBD}), the steady state vector $\mathbf{x}_i$ satisfies 
		\begin{equation}
		\mathbf{x}_i\mathbf{P}_i = \mathbf{x}_i,\; \mathbf{x}_i\mathbf{1}_{r_i} = 1,
		\end{equation}
		which is in the form of $\mathbf{A}\mathbf{x}=\mathbf{b}$. Employing \cite[Lemma 1]{Krikidis2012}, the lemma can be proved. 
\end{proof}

\subsection{Alternative Computationally Convenient Solution }

The mathematical complexity required for the inversion in (\ref{eq:DirectSol_SSP}) can be cumbersome, specially for large number of priority classes and large queue sizes $k_i$. Thus, a less-complex and mathematically tractable solution is sought. To this end, the \ac{MAM} is a powerful mathematical tool which is most suited to Markov chains with QBD structure \cite{G.Kulkarni1999},\cite{S.Alfa2015}. Based on the state transition matrix defined in (\ref{eq:QBD}), the following lemma derives the steady state distribution for the $i$-th priority queue. 
\begin{lemma}\label{lem:SSP_MAM}
	The steady state distribution based on the \ac{MAM} for the $i$-th queue is 
	\begin{equation}\label{eq:MAM_SSP}
	\mathbf{x}_{i,j}= \begin{cases}
	\mathbf{\Upsilon}_i\mathbf{A}_{2,i} \big(\mathbf{I}_{m_i} - \mathbf{B}_{1,i} \big)^{-1}, & j = 0,  \\
	\mathbf{\Upsilon}_i,  & j =1, \\
	\mathbf{x}_{i,1}\mathbf{R}_i^{i-1}, &  j > 1,
	\end{cases}
	\end{equation}
	where $\mathbf{\Upsilon}_i=\mathbf{x}_{i,0}\mathbf{C}_{i}\big(\mathbf{I}_{m_i}-\mathbf{A}_{1,i}-\mathbf{R}_i\mathbf{A}_{2,i} \big)^{-1}$ and $\mathbf{R}_i = \mathbf{A}_{0,i}(\mathbf{I}_{m_i}-\mathbf{A}_{1,i}-\mathbf{A}_{0,i} \mathbf{1}_{m_i}\bm{\beta}_i\big)^{-1}$ is the \ac{MAM} matrix. In addition, (\ref{eq:MAM_SSP}) must  satisfy the  normalization $\mathbf{x}_{i,0}\mathbf{1}_{m_i}+ \mathbf{\Upsilon}_i(\mathbf{I}_{m_i} - \mathbf{R}_i)^{-1}\mathbf{1}_{m_i} = 1$.
\end{lemma}
\begin{proof}
	Based on \cite{G.Kulkarni1999},\cite{S.Alfa2015}, $\mathbf{R}_i$ is the minimal non-negative solution to the quadratic equation $\mathbf{R}_i = \mathbf{A}_{0,i} + \mathbf{A}_{1,i} + \mathbf{A}_{2,i}$. Let $\mathbf{x}_{i,0}$ and $\mathbf{x}_{i,1}$ be the solution to 
	\begin{equation}\label{eq:xox1_MAM}
	\Big[\mathbf{x}_{i,0} \;\; \mathbf{x}_{i,1}\Big] = \Big[\mathbf{x}_{i,0} \;\; \mathbf{x}_{i,1}\Big] \left[ \begin{array}{ll} \mathbf{B}_{1,i} & \mathbf{C}_{i} \\ \mathbf{A}_{2,i} & \mathbf{A}_{1,i}+\mathbf{R}_i\mathbf{A}_{2,i} \end{array}\right].
	\end{equation}
	The employed \ac{DTMC}  has an advantageous feature that can be exploited, since $\mathbf{A}_{2,i}$ is a rank one matrix, which simplifies $\mathbf{R}_i$ to $\mathbf{R}_i = \alpha_i\mathbf{S}_{i}(\mathbf{I}_{m_i}-\bar{\alpha_i}\mathbf{s}_i\bm{\beta}_i - \bar{\alpha_i}\mathbf{S}_{i}-\alpha_i\mathbf{S}_{i}\mathbf{1}_{m_i}\bm{\beta}_i).$ Given that (\ref{eq:stability}) is satisfied, $\mathbf{R}_i$ has a spectral radius less than one \cite{S.Alfa2015}. The solution to (\ref{eq:xox1_MAM}) is
	\ifCLASSOPTIONdraftcls
	\begin{align}
	\mathbf{x}_{i,0}&= \alpha_i \mathbf{x}_{i,0} \mathbf{S}_{0}(\mathbf{I}_{m_i} - \alpha_i\mathbf{s}_i \bm{\beta}_i-\bar{\alpha}_i\mathbf{S}_i - \mathbf{R}_i \bar{\alpha}_i \mathbf{s}_i \bm{\beta}_i)^{-1} \bar{\alpha}_i \mathbf{s}_i\bm{\beta}_i(\mathbf{I}_{m_i}-\bar{\alpha}_i \mathbf{S}_{0,i})^{-1}, 
	\end{align}
	\else
	\begin{align}
	\mathbf{x}_{i,0}&= \alpha_i \mathbf{x}_{i,0} \mathbf{S}_{0}(\mathbf{I}_{m_i} - \alpha_i\mathbf{s}_i \bm{\beta}_i-\bar{\alpha}_i\mathbf{S}_i - \mathbf{R}_i \bar{\alpha}_i \mathbf{s}_i \bm{\beta}_i)^{-1} \\ \nonumber
	&\bar{\alpha}_i \mathbf{s}_i\bm{\beta}_i(\mathbf{I}_{m_i}-\bar{\alpha}_i \mathbf{S}_{0,i})^{-1}, 
	\end{align}
	\fi 	
	with the normalization $\mathbf{x}_{i,0}\mathbf{1}_{m_i} + \alpha_i \mathbf{x}_{i,0} \mathbf{S}_{0}(\mathbf{I}_{m_i} - \alpha_i\mathbf{s}_i \bm{\beta}_i-\bar{\alpha}_i\mathbf{S}_i - \mathbf{R}_i \bar{\alpha}_i \mathbf{s}_i \bm{\beta}_i)^{-1} (\mathbf{I}_{m_i} - \mathbf{R}_i)^{-1}\mathbf{1}_{m_i} = 1$. Finally, $\mathbf{x}_{i,1}$ is obtained through solving (\ref{eq:xox1_MAM}) and $\mathbf{x}_{i,j} = \mathbf{x}_{i,1}\mathbf{R}^{-1}$. Substituting the component matrices, the lemma is reached. 
\end{proof}

\subsection{Vacation Model Verification}
As verification, the proposed vacation-based preemptive model is compared against the conventional method presented in \cite[Chapter 9]{S.Alfa2015} for the case of $N=2$. Assuming a hypothetical fixed service probability $p_i$, Fig. 4 compares the conventional method with the proposed one. It is observed that the vacation-based model exactly characterizes the priority queues evolution while offering a computationally convenient, tractable, and scalable model for larger number of priority classes, whereas for higher values of $N$, the conventional method becomes highly complex.

\begin{figure} 
	\centering
	\ifCLASSOPTIONdraftcls
	\begin{tikzpicture}[scale=0.9]
\begin{groupplot}[group style={
	group name=myplot,
	group size= 2 by 1,  horizontal sep=2cm}, height=1.4in,width=2in]
\nextgroupplot[title={{(a)}},
scale only axis,
xmin=0,
xmax=6,
xlabel style={font=\color{white!15!black}},
xlabel={$j$},
xtick={0,1,2,3,4,5,6}, 
ymin=0,
ymax=1,
ylabel style={font=\color{white!15!black}},
ylabel={$x_{1,j}$},
axis background/.style={fill=white},
xmajorgrids,
ymajorgrids]
\addplot[ycomb, color=red, dashed, line width=1.0pt, mark=otimes, mark options={solid, fill=red, red}] table[row sep=crcr] {%
	0	0.800000301068341\\
	1	0.177777844681854\\
	2	0.0197530938535393\\
	3	0.00219478820594881\\
	4	0.000243865356216534\\
	5	2.7096150690726e-05\\
	6	3.01068341008067e-06\\
};

\addplot[ycomb, color=blue, dotted, line width=2.0pt, mark=x, mark options={solid, blue}] table[row sep=crcr] {%
	0	0.800000301068341\\
	1	0.177777844681854\\
	2	0.0197530938535393\\
	3	0.00219478820594881\\
	4	0.000243865356216534\\
	5	2.70961506907261e-05\\
	6	3.01068341008067e-06\\
};

\nextgroupplot[title={{\smash{(b)}}},
scale only axis,
scale only axis,
xmin=0,
xmax=5,
xlabel style={font=\color{white!15!black}},
xlabel={$j$},
xtick={0,1,2,3,4,5}, 
ylabel={$x_{2,j}$},
ymin=0,
ymax=1,
axis background/.style={fill=white},
xmajorgrids,
ymajorgrids]
\addplot[ycomb, color=red, dashed, line width=2.0pt, mark=*, mark options={solid, fill=red, red}] table[row sep=crcr] {%
	0	0.025105333738797\\
	1	0.0619812817198168\\
	2	0.0941059652153034\\
	3	0.146851464361824\\
	4	0.230227496730046\\
	5	0.441728458234214\\
};\addlegendentry{\cite{Emara2020}}

\addplot[ycomb, color=blue, dotted, line width=2.0pt, mark=x, mark options={solid, fill=blue, blue}] table[row sep=crcr] {%
	0	0.025105333738797\\
	1	0.0619812817198169\\
	2	0.0941059652153035\\
	3	0.146851464361824\\
	4	0.230227496730046\\
	5	0.441728458234213\\
};
\addlegendentry{Proposed}

\end{groupplot}

\end{tikzpicture}
	\vspace{-0.15in}
	\else
	\begin{tikzpicture}[scale=0.9]
\begin{groupplot}[group style={
	group name=myplot,
	group size= 2 by 1,  horizontal sep=1.5cm}, height=1.8in,width=1.3in]
\nextgroupplot[title={{(a)}},
scale only axis,
xmin=0,
xmax=6,
xlabel style={font=\color{white!15!black}},
xlabel={$j$},
xtick={0,1,2,3,4,5,6}, 
ymin=0,
ymax=1,
ylabel style={font=\color{white!15!black}},
ylabel={$x_{1,j}$},
axis background/.style={fill=white},
xmajorgrids,
ymajorgrids]
\addplot[ycomb, color=red, dashed, line width=1.0pt, mark=otimes, mark options={solid, fill=red, red}] table[row sep=crcr] {%
	0	0.800000301068341\\
	1	0.177777844681854\\
	2	0.0197530938535393\\
	3	0.00219478820594881\\
	4	0.000243865356216534\\
	5	2.7096150690726e-05\\
	6	3.01068341008067e-06\\
};

\addplot[ycomb, color=black, dotted, line width=1.0pt, mark=x, mark options={solid, black}] table[row sep=crcr] {%
	0	0.800000301068341\\
	1	0.177777844681854\\
	2	0.0197530938535393\\
	3	0.00219478820594881\\
	4	0.000243865356216534\\
	5	2.70961506907261e-05\\
	6	3.01068341008067e-06\\
};

\nextgroupplot[title={{\smash{(b)}}},
scale only axis,
scale only axis,
xmin=0,
xmax=5,
xlabel style={font=\color{white!15!black}},
xlabel={$j$},
xtick={0,1,2,3,4,5}, 
ylabel={$x_{2,j}$},
ymin=0,
ymax=1,
axis background/.style={fill=white},
xmajorgrids,
ymajorgrids]
\addplot[ycomb, color=red, dashed, line width=1.0pt, mark=*, mark options={solid, fill=red, red}] table[row sep=crcr] {%
	0	0.025105333738797\\
	1	0.0619812817198168\\
	2	0.0941059652153034\\
	3	0.146851464361824\\
	4	0.230227496730046\\
	5	0.441728458234214\\
};\addlegendentry{\cite[Ch. 9]{S.Alfa2015}}

\addplot[ycomb, color=black, dotted, line width=1.0pt, mark=x, mark options={solid, fill=black, black}] table[row sep=crcr] {%
	0	0.025105333738797\\
	1	0.0619812817198169\\
	2	0.0941059652153035\\
	3	0.146851464361824\\
	4	0.230227496730046\\
	5	0.441728458234213\\
};
\addlegendentry{Proposed}

\end{groupplot}

\end{tikzpicture}
	\fi 		
	\caption{Steady state probabilities for (a) first (b) second priority class for $k_1 = 6, k_2 = 5, \alpha_1 = 0.1, \alpha_2 = 0.5 \& p_1 = p_2 = 0.5$.}
	\label{fig:VacationVerif}
\end{figure}

It is clear that in order to compute the steady state distributions $x_{i,j}$ of the $i$-th queue, one need to compute $p_i$. Such inter-dependency highlights the interaction between the  microscopic and macroscopic scales in the network. In what follows, we present the framework adopted to characterize $p_i$ based on stochastic geometry analysis. 

\thispagestyle{empty}
\section{Macroscopic Large Scale Analysis}\label{sec:sg_analysis}

Based on \eqref{eq:SINR_1}, it is clear that $p_i$ is a function of the aggregate network mutual interference induced by the macroscopic interactions between the devices.  This section utilizes stochastic geometry to delve into the network-wide interactions between devices and characterizes the coverage probability defined in \eqref{eq:SINR_1}. Before proceeding further, we state two commonly used and core approximations that are utilized in this work for tractability and mathematical convenience. 

\begin{approximation} \label{approx1}
(i) The spatial correlations between adjacent Voronoi cell areas are ignored. (ii) All devices in the network are assumed to perform (i.e., in terms of coverage probability) as the typical device located at the origin. Both approximations are validated in Section~\ref{sec:simulation_results} against independent Monte Carlo simulations.
\end{approximation} 

\begin{remark} 
(i) Implies that all devices will have independent and identically distributed transmit powers to invert their path-loss to the serving BS.  Such assumption is commonly used and verified in the literature \cite{Lee2014, ElSawy2014}. (ii) For static networks, the coverage probability is location dependent, which is captured via the meta distribution~\cite{Haenggi_meta} and can be incorporated to the spatiotemporal analysis as in~\cite{Chisci2019, Yang2019}. However, it is shown \cite{Gharbieh2017, Gharbieh2018, Elsawy_meta, Haenggi_meta2} that such location dependence diminishes with path-loss inversion and random channel selection.
\end{remark}

Exploiting Approximation~\ref{approx1}(i) and \ref{approx1}(ii), the coverage probability of an $i$-th priority packet transmitted from a typical device located at the origin can be further expressed as
\begin{equation}\label{eq:TSP_1}
	p_i = \mathbb{P}\{\text{SINR}_i > \theta\} = \mathbb{P}\Big\{\frac{\rho h_o}{I_i + \sigma^2} \ge \theta\Big\},
\end{equation}
where $h_o$ is the channel gain between the device and its serving \ac{BS}, $\sigma^2$ is the noise power, and $I_i$ is the aggregate interference seen by an $i$-th priority packet. In details, $I_i$ is expressed as  
\begin{equation}
I_i = \sum_{y_j \in \mathrm{\Phi}\setminus y_o} \mathbbm{1}_{\{a_{i,j}\}} P_j g_j ||y_j - z_o||^{-\eta},
\end{equation}
where $y_j$ is the location of an interfering device (all active devices will be interfering except  the typical device $ \mathrm{\Phi}\setminus y_o$), $a_{i,j}$ is the event that the device located at $y_j$ is transmitting on the same channel as the typical device, $P_j$ is its transmit power, $g_j$ is the channel power gain between the interfering device and the serving BS,  $||.||$ is the Euclidean norm, and $z_o$ is the typical device's serving \ac{BS}'s location.

\begin{remark}
It is worth noting that $p_i$ across different priority classes will only be different for the dedicated channel allocation, where the channel selection is dependent on the packet priority. Hence, a device sending a packet of priority $i$ may only experience interference form devices transmitting packets of the  same priority. However, for the case of shared channel allocation, the coverage probability is agnostic the packets priorities. 
\end{remark}

Due to the assumed exponential distribution of $h_o$, the channel inversion power control and the definition of the \ac{LT}, (\ref{eq:TSP_1}) can be expressed as 
\begin{equation}\label{eq:TSP_2}
p_i = \text{exp}\Big \{ -\frac{\sigma^2\theta}{\rho} \Big\} \mathcal{L}_{I_i}\Big (\frac{\theta}{\rho} \Big),
\end{equation}
where $\mathcal{L}_{I_i}(.)$ is the \ac{LT} of the aggregate interference $I_{i}$. One can observe from (\ref{eq:TSP_2}) the effect of fading, power control, and \ac{SINR} threshold on the achieved transmission probabilities, which in return affects the queues temporal evolution. Thus, coupling the queues departure probabilities and the aggregate interference in the network. In the remaining of this section, we characterize the coverage probability for three different channel allocation strategies.

\subsection{Dedicated allocation} 
This scheme considers an orthogonal allocation among the active queues based on their priority. The interfering sources to an active transmission of the $i$-th priority queue can only be from the set of all active devices having packets to be transmitted in their $i$-th priority queue. The coverage probability of an $i$-th priority packet under the dedicated allocation is derived as follows.

\begin{theorem}\label{lem:dedicated}
	The coverage probability ${p}_i$ of a packet belonging to the $i$-th priority class under the dedicated allocation strategy is given by
	\begin{align}\label{eq:dedicated_2}
	{p}_i &\approx \; \;\frac{\exp \left\{-\frac{\sigma^{2} \theta}{\rho}-\frac{2 \theta \zeta_i \kappa_{i,\text{m}}}{(\eta-2)} {}_2F_{1}(1,1-2 / \eta, 2-2 / \eta,-\theta)\right\}}{\left(1+\frac{\theta \zeta_i \kappa_{i,\text{m}}}{(1+\theta) c}\right)^{c}}, \nonumber \\
	&\myeqb \; \;\frac{\exp \left\{-\frac{\sigma^{2} \theta}{\rho}- \zeta_i \kappa_{i,\text{m}} \sqrt{\theta} {\arctan}\Big(\sqrt{\theta}\Big)\right\}}{\left(1+\frac{\theta \zeta_i \kappa_{i,\text{m}}}{(1+\theta) c}\right)^{c}},
	\end{align}
	where $\zeta_i = \sum_{z_i=1}^{k_{i}} \mathbb{P}\{(0,0,\cdots, 0,z_i)\}$ is the joint probability of having no packets with priority higher than $i$ and at least a packet with priority $i$, $\kappa_{i,\text{m}} = \frac{\mu}{\lambda C_{i,\text{m}}}$ is the average number of devices per \ac{BS} per channel, where $\text{m}\in\{\text{EA, WA}\}$ indicates \ac{EA} or \ac{WA} dedication strategy. ${}_2F_{1}(\cdot)$ is the Gaussian hypergeometric function that is defined as $_2F_{1}(a,b,u;z) = \sum_{k=0}^{\infty}\frac{(a)_k(b)_kz^k}{(u)_kk!}$ and $c = 3.575$. The approximation is due to Approximation~\ref{approx1}(i) and the employed approximate PDF of the PPP Voronoi cell area in $\mathbb{R}^2$ as shown in (\ref{eq:laplace_2}).
\end{theorem}
\begin{proof}
	See Appendix \ref{se:Appendix_B}.
\end{proof}
The parameter $\kappa_i \zeta_i$ represents the portion of devices attempting a transmission of an $i$-th priority packet. Thus, interfering on the typical device that is attempting the transmission of its own $i$-th priority packet. Moreover, $\kappa_i$ is affected by the number of channels assigned to each priority class. Through this work, we investigate two dedicated channel allocation strategies; namely, \ac{EA} and \ac{WA}. The former equally splits the total available channels among the existing priority classes, whereas the latter considers an allocation of channels that is dependent on that given priority class arrival probability. Mathematically, the number of allocated channels for the equal and weighted schemes are expressed as
\begin{equation}
C_{i, \text{EA}} = \frac{C}{N},\;\; C_{i, \text{WA}} = C \frac{\alpha_i}{\sum_{j=1}^{N} \alpha_j}.
\end{equation}
\vspace{-12pt}
\subsection{Shared allocation } 
This strategy considers the case of inter-class channel multiplexing among all the active devices irrespective of the packet's priority that is to be transmitted. That is, all the active devices can mutually interfere regardless of the priority of the packets being transmitted. Hence, all the devices with non-empty queues are potential interferers to the typical device's packet.  Recalling the preemptive-based mechanism, the probability of being a potential interferer is the complement of the joint probability that all the $N$ priority queues are empty. In the following theorem, the coverage probability of an $i$-th priority packet under the shared allocation is derived. 
\begin{theorem}\label{lem:shared}
	The coverage probability ${p}_i$ of a packet belonging to the $i$-th priority class under the shared allocation strategy is given by
	\begin{align}\label{eq:shared}
	&p_{i} \approx  \frac{\exp \left\{-\frac{\sigma^{2} \theta}{\rho}-\frac{2 \theta \bar{\delta} \kappa}{(\eta-2)} {}_2F_{1}(1,1-2 / \eta, 2-2 / \eta,-\theta)\right\}}{\left(1+\frac{\theta \bar{\delta} \kappa}{(1+\theta) c}\right)^{c}}, \nonumber \\
	 &\;\; \myeqb \; \;\frac{\exp \left\{-\frac{\sigma^{2} \theta}{\rho}- \bar{\delta} \kappa \sqrt{\theta} {\arctan}\Big(\sqrt{\theta}\Big)\right\}}{\left(1+\frac{\theta \bar{\delta} \kappa}{(1+\theta) c}\right)^{c}},
	\end{align}
	where $\delta=\mathbb{P}\{(0,0,\cdots, 0,0)\}$ is the joint probability of having no packets in all $N$ priority queues,  $\kappa=\frac{\mu}{\lambda C}$ is the average number of devices per \ac{BS} per channel and $c = 3.575$. The approximation is due to Approximation~\ref{approx1}(i) and the employed approximate PDF of the PPP Voronoi cell area in $\mathbb{R}^2$ as shown in (\ref{eq:laplace_2}).
\end{theorem}
\begin{proof}
	Since all the packets being transmitted experience the same aggregate interference under the shared allocation, $p_i$ of all the queues are identical. Furthermore, a device is attempting a transmission if it has any packets within its $N$ priority queues. Thus, the portion of interfering devices within the network is $\mu\bar{\delta}$, where $\delta$ is the joint probability of having no packets in all the $N$ priority queues. Finally, the theorem is realized following similar steps as Theorem 1. 
\end{proof}

In summary, the shared channel allocation strategy aims at allowing the devices to utilize all available channels. Thus, a given device will have a larger pool of channels to utilize for its transmission, while experiencing mutual interference from different priority transmission. On the other hand, the dedicated strategies provides a limited number of the channels for a given class, based on an allocation criteria, either equally or proportionally. This prohibits mutual interference from different priority transmission.

\subsection{Iterative Solution}

As discussed in Section \ref{sec:QT_anaylsis}, the idle probability of an $i$-th priority queue employed at a given \ac{IoT} device governs the interference it causes within the network. In addition, the aggregate network interference affects the idle probability of each device. Thus, an inter-dependency exists between the devices activity and aggregate interference scales. Such inter-dependency can be solved iteratively as presented in Algorithm 1, which converges uniquely to a solution by virtue of the fixed point theorem \cite{Zhou2016}. Regarding the complexity, the dedicated allocation scheme is considered to be more complex compared to the shared one, as it requires an additional coordination step to compute the portion of channels available to each priority class. In order to conduct this, prior knowledge of the number of priority classes or the arrival probabilities $\alpha_i$ are required for the EA and WA strategies, respectively. On the other hand, the shared allocation alleviates such step, as all the channels are available irrespective of the packet's priority class.

\begin{figure}
	\begin{algorithm}[H]\label{alg:iterative}
		\setstretch{1.0}
		\caption{Iterative computation of $p_i$ and $\mathbf{x}_{i}$ for dedicated and shared channel allocation}\label{euclid}
		\begin{algorithmic}
		\Procedure{}{$(\alpha_1 \; \alpha_2 \; \cdots \; \alpha_N), \lambda, \mu, \eta, \theta, C, \epsilon$}  \Comment{$\epsilon$ is a convergence tolerance parameter}
		\LState initialize  $\delta$ and $\zeta_i \in[0,1]$ 
		\While {$||\mathbf{x}_{i}^k - \mathbf{x}_{i}^{k-1}|| \geq \epsilon$}   
		\LState Compute $p_i$ from (\ref{eq:dedicated_2})-dedicated or (\ref{eq:shared})-shared. 					 
		\LState Construct $\mathbf{S}_{0,i}$ and $\mathbf{S}_{i}$ from Proposition \ref{lem:S_So}.
		\If {$\bm{\pi}_i\mathbf{A}_{2,i} \mathbf{1} > \bm{\pi}_i\mathbf{A}_{0,i} \mathbf{1}$} \Comment{non-overflow (i.e., stability) condition}
		\LState Solve $\mathbf{x}_i$ based on Lemma \ref{lem:SSP} or Lemma \ref{lem:SSP_MAM}.
		\LState Compute $\delta$ and $\zeta_i$ from $\mathbf{x}_i$.
		\LState Compute $p_i$ from   (\ref{eq:dedicated_2})-dedicated or (\ref{eq:shared})-shared. 
		\Else
		\LState Set $\delta = 0$ and calculate $\zeta_i$. 
		\LState Compute $p_i$ from   (\ref{eq:dedicated_2})-dedicated or (\ref{eq:shared})-shared. 
		\LState Break.
		\EndIf
		\LState Increment k.
		\EndWhile
		\LState \Return  $p_i$ and $\mathbf{x}_{i} \; \forall i $.
		\EndProcedure
		\end{algorithmic}
	\end{algorithm}
\end{figure}

\subsection{Performance Metrics}
Based on the provided iterative framework, once can evaluate the steady state distribution of the $N$ priority queues. To this end, a number of \acp{KPI} can be evaluated, which are insightful when designing and assessing massive \ac{PMT} \ac{IoT} networks. First, the \ac{TSP} is evaluated as $\text{TSP}_i = \gamma_i p_{i}$,
where $\gamma_i$, defined in \eqref{gggg}, is the probability that the sever is available to serve the $i$-th priority packet. Articulated differently, $\gamma_i$ is the probability that all higher priority queues are empty such that the device is able to send an $i$-th priority packet. Such transmission attempt succeeds with probability $p_i$ as given by \eqref{eq:dedicated_2} for the dedicated allocation and \eqref{eq:shared} for shared allocation. Let $Q_i$ be the instantaneous number of packets at the $i$-th queue, then the average number of packets is
\begin{equation}
\mathbb{E}\big\{Q_{i}\big\}=\sum_{n=1}^{k_i} n \mathbb{P}\left\{Q_{i}=n\right\}=\sum_{n=1}^{k_i} n  x_{i, n}. 
\end{equation}
For the $i$-th priority packet, its transmission will be postponed till all the packets belonging to higher classes are successfully served. Transmission availability for the $i$-th priority class in a generic device denotes the probability that the $i$-th priority queue is non-empty and that all higher priority queues are empty. Thus, transmission availability is evaluated as 
\begin{equation}
\mathcal{A}_i = 1 - \sum_{j = 1}^{i-1} \sum_{m_j = 1}^{k_j} x_{j,m_j}.
\end{equation}
A critical \ac{KPI} in prioritized traffic is the information freshness, which is quantified via the age of information \cite{Kaul2018}. Specifically, we focus in our work on the \ac{PAoI}, which is defined as as the value of information age resulted immediately prior to receiving a given packet \cite{Huang2015}.\footnote{Focus on the \ac{PAoI} stems from its importance in analyzing guaranteed performance for time critical applications.} Mathematically, the \ac{PAoI} is defined as	
\begin{equation}\label{PAoI}
	\text{PAoI}_i = \mathbb{E}\big\{I_{i}\big\} + \mathbb{E}\big\{W_{i}\big\} + \mathbb{E}\big\{D_{i}\big\},
	\end{equation}
where $\mathbb{E}\big[W_{i}\big], \mathbb{E}\big[D_{i}\big]$ and $\mathbb{E}\big[I_{i}\big]$ denote the average queueing delay, average 	
transmission delay and inter-arrival delay, respectively. Based on the adopted geometric distribution for packets arrival, the average inter-arrival times simplifies to $\mathbb{E}\big[I_{i}\big] = \frac{1}{\alpha_i}$. In addition, let $W_i$ be the queueing delay (i.e., number of time slots spent in the queue before the service of the $i$-th priority queue starts) for a randomly selected packet, then the average queueing delay is given by 
\begin{equation}\label{eq_queuing_delay}
\mathbb{E}\big\{W_{i}\big\}=\sum_{n=0}^{\infty} n \mathbb{P}\{W_{i}=n\}, 
\end{equation}
where the temporal distribution of the delay (i.e., across different packets) can be obtained as $\mathbb{P}\left\{W_{i}=0\right\} = x_{i,0}$ and $\mathbb{P}\{W_{i}=j\} = \sum_{k=1}^{j} \mathbf{x}_{i,k}\mathbf{G}_j^{(k)}\mathbf{1}$, where $\mathbf{G}_{i,j}^{(k)}$ represents the probability of having $k$ packets in the $i$-th priority queue and being serviced in $j$ time slots with
\begin{equation}
\mathbf{G}_{i,j}^{(k)}=\left\{\begin{array}{ll}
{\mathbf{S}_i^{j-1} \mathbf{s}_i\bm{\beta}_i} & {k=1,} \\
{(\mathbf{s}_i \bm{\beta}_i)^{k}} & {j=k, k \geq 1,} \\
{\mathbf{S}_i \mathbf{G}_{i,j-1}^{(k)}+\mathbf{s} \beta \mathbf{G}_{i,j-1}^{(k-1)}} & {k \geq j \geq 1.}
\end{array}\right.
\end{equation}
Based on the considered PH type distribution for the vacation duration, let  $W_i$ be the number of time slots spent in the queue before the service starts for a randomly chosen packet. Averaging over all packets, the transmission delay can be computed as \cite[Section 2.5.3]{S.Alfa2015}
\begin{equation}\label{eq_local_delay}
\mathbb{E}\big\{D_{i}\big\}= \bm{\beta}_i(\mathbf{I}_{m_i} - \mathbf{S}_i)^{-1}. 
\end{equation}
Finally, the PAoI is evaluated by plugging (\ref{eq_queuing_delay}) and (\ref{eq_local_delay}) into (\ref{PAoI}). 
\thispagestyle{empty}
% !TEX root =../Integration.tex
%%%%%%%%%%%%%%%%%%%%%%%%%%%%%%%%%%%%%%%%%%%%%%%%%%%%%%%%%%%%%%%%%%%%%%%%%%%%%%%%%%
\section{Simulation Results}\label{sec:simulation_results}
Through this section various numerical results are presented that aim at (a) validating the proposed analytical model, (b) highlighting the influence of the different channel allocation strategies, and (c) showing priority-aware wireless-based system design insights. 

\subsection{Simulation Methodology}

The developed simulation framework  incorporates microscopic and macroscopic averaging, where the former addresses the steady state temporal statistics of the different queues employed at each device and the latter addresses the stochastic geometric network-wide performance. The simulation area is $10 \times 10 \text{ km}^2$ with a wrapped-around boundaries to ensure unbiased statistics imposed by the network boundary devices. Unless otherwise stated, we consider the following physical layer parameters: $\kappa = 1$ devices/BS/channel, $C = 64$ channels, $\eta = 4$ and $\rho=\sigma^2=-90$ dBm. For the MAC layer parameters, we consider three priority classes with $(\alpha_1, \alpha_2, \alpha_3 )=(0.1, 0.25, 0.35)$, where all the queues have equal size (i.e., $k_1=k_2=k_3=8$). The proposed priority-aware transmission schemes are compared to a reference multi-stream \ac{PA} \ac{FCFS} queueing model. In such model, the transmission is granted on an \ac{FCFS} basis, equally among  the all existing $N$ priority classes.

\begin{figure*}
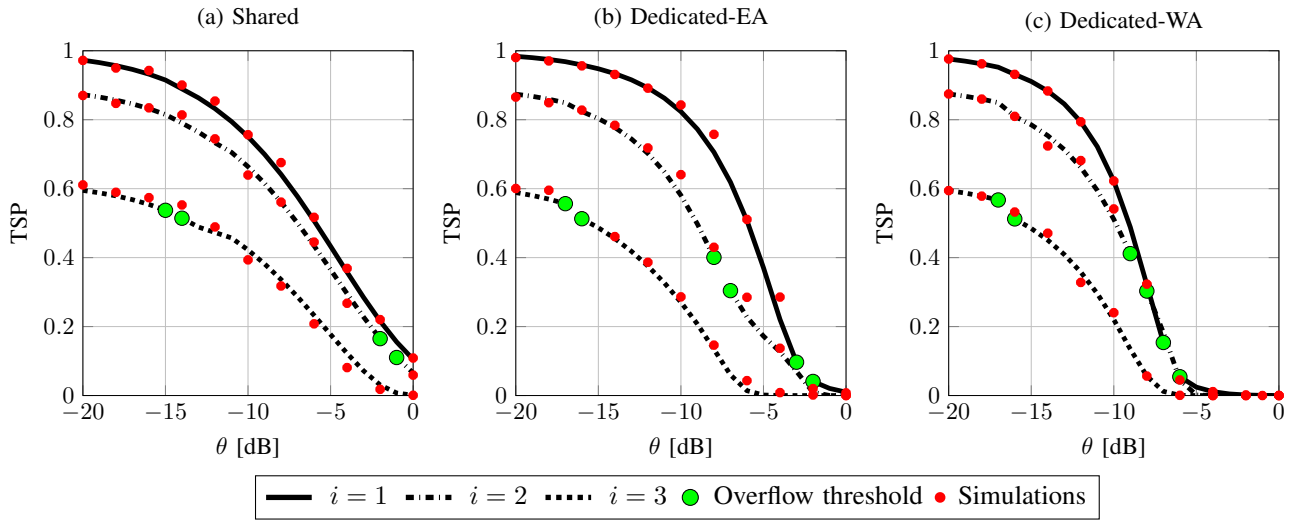

	\centering
	\ifCLASSOPTIONdraftcls
	\input{simulation_results/figures/TSP/1col/TSP.tex}
	\vspace{-0.1in}
	\else
	\input{simulation_results/figures/TSP/2col/TSP.tex}
	\fi
	\caption{TSP for three priority classes as a function of the \ac{SINR} threshold $\theta$  under (a) Shared (b) Dedicated-equal allocation (c) Dedicated-weighted allocation.}
	\label{fig:TSP} 
\end{figure*}

Synchronous time-slotted system is adopted and each microscopic simulation run is considered as a time slot where independent channel gains are instantiated and packets are generated probabilistically. The queue occupancy for each of the considered priority classes are tracked. For a transition from one time slot to another, packets are independently generated at every device for all queues based on the batch arrival process (i.e., $\alpha_i$). Every device with a non-empty queue of the $N$ queues tries to communicate its backlogged packets with its serving \ac{BS} based on the employed preemptive priority-aware transmission strategy. For a device with non-empty $i$-th priority queue, a packet is dispatched from the $i$-th priority queue if and only if i) all higher priority queues are empty, and ii) the achieved uplink \ac{SINR}$_i$ on the selected channel is greater than $\theta$. In order to ensure that the different queues at the devices are in steady state,  simulation is first initiated with all queues at the devices as being idle and then it runs for a sufficiently high number of time slots until the steady-state is reached. Let $\hat{\mathbf{x}}^k = [\mathbf{x}_{1,0}^k, \mathbf{x}_{2,0}^k, \cdots, \mathbf{x}_{N,0}^k]$ denotes the idle steady state probability  for the $t$-th iteration of the $N$ queues. Mathematically, the steady state is realized once $||\hat{\mathbf{x}}^k-\hat{\mathbf{x}}^{k-1}|| < \epsilon$, where $\epsilon$ is some predetermined tolerance. After steady state is reached, all temporal statistics are then gathered based on sufficiently large number of microscopic realizations. Finally, the whole process is repeated for sufficiently large number of macroscopic network realizations to ensure spatial ergodicity is reached. 

\subsection{Performance Evaluation}

We start with the framework validation for all considered priority classes and proposed channel allocation strategies. Fig. \ref{fig:TSP} shows the\ac{TSP} for three priority classes against the \ac{SINR} threshold $\theta$. The close matching between the theoretical and simulation results validates the developed spatiotemporal mathematical model. Moreover, focusing on a given channel allocation strategy and a priority class for low values of $\theta$, the devices are able to empty their queues and go into idle state when operating below the overflow threshold. This leads to a lower network aggregate interference. As $\theta$ increases, the transmission success probability decreases, which leads in turn into having higher aggregate network interference. Based on the prioritized transmission and the assumption that $\alpha_j > \alpha_i,\; \forall j > i$, it is expected that $\text{TSP}_j < \text{TSP}_i$. This is justified as lower priority packets are served only if all the higher priority queues are empty. In addition, it is clear that the SINR threshold $\theta$, at which the system transitions from non-overflow to overflow operation depends on the priority class. Note that the overflow thresholds depict the point where the probability of queues overflow starts to dominate.

\begin{figure}
	\centering
	\ifCLASSOPTIONdraftcls
	% This file was created by matlab2tikz.
%
%The latest updates can be retrieved from
%  http://www.mathworks.com/matlabcentral/fileexchange/22022-matlab2tikz-matlab2tikz
%where you can also make suggestions and rate matlab2tikz.

\definecolor{mycolor1}{rgb}{0.4, 1.0, 0.0}
\definecolor{mycolor2}{rgb}{0.03, 0.91, 0.87}
\definecolor{mycolor3}{rgb}{1,0.0,1.0}
\definecolor{mycolor4}{rgb}{0.69, 0.4, 0.0}%
\definecolor{mycolor5}{rgb}{0.97690,0.98390,0.08050}%

\begin{tikzpicture}[scale=1]
\begin{groupplot}[group style={
	group name=myplot,
	group size= 3 by 1, horizontal sep=1cm},height=1.5 in, width= 0.35*\columnwidth, xtick={-18,-14,-10,-6,-2},]
\nextgroupplot[title={(a) First priority class}, align=left,
major x tick style = transparent,
ylabel={TSP},
xlabel={$\theta$ [dB]},
ybar=4*\pgflinewidth,
bar width=9pt,
ymajorgrids = true,
%symbolic x coords={50,70,90,110,130},
scaled y ticks = false,
ymin=0, 
ymax=1,
ytick={0,0.2,0.4,0.6,0.8,1},
legend to name=grouplegend1,
legend cell align=left,
legend style={
	legend columns=4,fill=none,draw=black,anchor=center,align=left, column sep=0.13cm}]

\addlegendimage{style={mycolor1,fill=mycolor1,mark=none}}
\addlegendentry{Shared}
\addlegendimage{style={mycolor2,fill=mycolor2,mark=none}}
\addlegendentry{Ded-EA}
\addlegendimage{style={mycolor3,fill=mycolor3,mark=none}}
\addlegendentry{Ded-WA}
\addlegendimage{style={red,fill=red,mark=none}}
\addlegendentry{PA}

\addplot[ybar, bar width=0.356, fill=mycolor1, draw=black, area legend] table[row sep=crcr] {%
	-18	0.9569\\
	-14	0.8894\\
	-10	0.7508\\
	-6	0.5068\\
	-2	0.216\\
};
\addplot[ybar, bar width=0.356, fill=mycolor2, draw=black, area legend] table[row sep=crcr] {%
	-18	0.9748\\
	-14	0.9337\\
	-10	0.8235\\
	-6	0.5073\\
	-2	0.04092\\
};
\addplot[ybar, bar width=0.356, fill=mycolor3, draw=black, area legend] table[row sep=crcr] {%
	-18	0.9621\\
	-14	0.8833\\
	-10	0.6224\\
	-6	0.05458\\
	-2	0.00165\\
};

\addplot[ybar, bar width=0.356, fill=red, draw=black, area legend] table[row sep=crcr] {%
	-18	0.540335800417419\\
	-14	0.447259870571471\\
	-10	0.266645893042116\\
	-6	0.116312048572885\\
	-2	0.0421205769842131\\
};

\coordinate (c1) at (rel axis cs:0.5,0);

\nextgroupplot[title={(b) Second priority class}, align=left,
major x tick style = transparent,
xlabel={$\theta$ [dB]},
ybar=4*\pgflinewidth,
bar width=9pt,
ymajorgrids = true,
scaled y ticks = false,
ymin=0, 
ymax=1,
ytick={0,0.2,0.4,0.6,0.8,1}]

\addplot[ybar, bar width=0.356, fill=mycolor1, draw=black, area legend] table[row sep=crcr] {%
	-18	0.8569\\
	-14	0.7901\\
	-10	0.6638\\
	-6	0.4342\\
	-2	0.165\\
};
\addplot[ybar, bar width=0.356, fill=mycolor2, draw=black, area legend] table[row sep=crcr] {%
	-18	0.8595\\
	-14	0.7774\\
	-10	0.5807\\
	-6	0.2275\\
	-2	0.009\\
};
\addplot[ybar, bar width=0.356, fill=mycolor3, draw=black, area legend] table[row sep=crcr] {%
	-18	0.8598\\
	-14	0.7539\\
	-10	0.5112\\	
	-6	0.0452\\
	-2	0\\
};

\addplot[ybar, bar width=0.356, fill=red, draw=black, area legend] table[row sep=crcr] {%
	-18	0.629796297990205\\
	-14	0.546844614495926\\
	-10	0.376753073876106\\
	-6	0.197797679775553\\
	-2	0.0771908083186376\\
};

\coordinate (c2) at (rel axis cs:0.5,0);

\nextgroupplot[title={(c) Third priority class}, align=left,
major x tick style = transparent,
xlabel={$\theta$ [dB]},
ybar=4*\pgflinewidth,
bar width=9pt,
ymajorgrids = true,
scaled y ticks = false,
ymin=0, 
ymax=1,
ytick={0,0.2,0.4,0.6,0.8,1}]

\addplot[ybar, bar width=0.356, fill=mycolor1, draw=black, area legend] table[row sep=crcr] {%
	-18	0.5791\\
	-14	0.5143\\
	-10	0.4223\\
	-6	0.2326\\
	-2	0.0307\\
};
\addplot[ybar, bar width=0.356, fill=mycolor2, draw=black, area legend] table[row sep=crcr] {%
	-18 0.56969\\
	-14	0.4454\\
	-10	0.2708\\
	-6	0.01406\\
	-2	0\\
};
\addplot[ybar, bar width=0.356, fill=mycolor3, draw=black, area legend] table[row sep=crcr] {%
	-18	0.5783\\
	-14	0.4507\\
	-10	0.2202\\	
	-6  0.00063\\
	-2	0\\
};

\addplot[ybar, bar width=0.356, fill=red, draw=black, area legend] table[row sep=crcr] {%
	-18	0.688979930973781\\
	-14	0.613850977475116\\
	-10	0.458921623893824\\	
	-6  0.26642354973061\\
	-2	0.109568406578523\\
};

\end{groupplot}

\coordinate (c3) at ($(c1)!1.05!(c2)$);
\node[below] at (c3 |- current bounding box.south)
{\ref{grouplegend1}};
\end{tikzpicture}
	\vspace{-0.15in}
	\else
	% This file was created by matlab2tikz.
%
%The latest updates can be retrieved from
%  http://www.mathworks.com/matlabcentral/fileexchange/22022-matlab2tikz-matlab2tikz
%where you can also make suggestions and rate matlab2tikz.

\definecolor{mycolor1}{RGB}{43, 140, 190}
\definecolor{mycolor2}{RGB}{166, 189, 219}
\definecolor{mycolor3}{RGB}{236, 231, 242}
\definecolor{mycolor4}{RGB}{168, 221, 181}
\definecolor{mycolor5}{rgb}{0.97690,0.98390,0.08050}%

\begin{tikzpicture}[scale=1]
\begin{groupplot}[group style={
	group name=myplot,
	group size= 1 by 3, vertical sep=1.5cm},height=4cm, width= 1*\columnwidth, xtick={-18,-14,-10,-6,-2},]
\nextgroupplot[title={(a) First priority class}, align=left,
major x tick style = transparent,
ylabel={TSP},
ybar=4*\pgflinewidth,
bar width=9pt,
ymajorgrids = true,
%symbolic x coords={50,70,90,110,130},
scaled y ticks = false,
ymin=0, 
ymax=1,
ytick={0,0.2,0.4,0.6,0.8,1},
legend to name=grouplegend1,
legend cell align=left,
legend style={
	legend columns=4,fill=none,draw=black,anchor=center,align=left, column sep=0.13cm}]

\addlegendimage{style={mycolor1,fill=mycolor1,mark=none}}
\addlegendentry{Shared}
\addlegendimage{style={mycolor2,fill=mycolor2,mark=none}}
\addlegendentry{Ded-EA}
\addlegendimage{style={mycolor3,fill=mycolor3,mark=none}}
\addlegendentry{Ded-WA}
\addlegendimage{style={mycolor4,fill=mycolor4,mark=none}}
\addlegendentry{PA}

\addplot[ybar, bar width=0.356, fill=mycolor1, draw=black, area legend] table[row sep=crcr] {%
	-18	0.9569\\
	-14	0.8894\\
	-10	0.7508\\
	-6	0.5068\\
	-2	0.216\\
};
\addplot[ybar, bar width=0.356, fill=mycolor2, draw=black, area legend] table[row sep=crcr] {%
	-18	0.9748\\
	-14	0.9337\\
	-10	0.8235\\
	-6	0.5073\\
	-2	0.04092\\
};
\addplot[ybar, bar width=0.356, fill=mycolor3, draw=black, area legend] table[row sep=crcr] {%
	-18	0.9621\\
	-14	0.8833\\
	-10	0.6224\\
	-6	0.05458\\
	-2	0.00165\\
};

\addplot[ybar, bar width=0.356, fill=mycolor4, draw=black, area legend] table[row sep=crcr] {%
	-18	0.540335800417419\\
	-14	0.447259870571471\\
	-10	0.266645893042116\\
	-6	0.116312048572885\\
	-2	0.0421205769842131\\
};

\coordinate (c1) at (rel axis cs:0.5,0);

\nextgroupplot[title={(b) Second priority class}, align=left,
major x tick style = transparent,
ybar=4*\pgflinewidth,
bar width=9pt,
ymajorgrids = true,
scaled y ticks = false,
ylabel={TSP},
ymin=0, 
ymax=1,
ytick={0,0.2,0.4,0.6,0.8,1}]

\addplot[ybar, bar width=0.356, fill=mycolor1, draw=black, area legend] table[row sep=crcr] {%
	-18	0.8569\\
	-14	0.7901\\
	-10	0.6638\\
	-6	0.4342\\
	-2	0.165\\
};
\addplot[ybar, bar width=0.356, fill=mycolor2, draw=black, area legend] table[row sep=crcr] {%
	-18	0.8595\\
	-14	0.7774\\
	-10	0.5807\\
	-6	0.2275\\
	-2	0.009\\
};
\addplot[ybar, bar width=0.356, fill=mycolor3, draw=black, area legend] table[row sep=crcr] {%
	-18	0.8598\\
	-14	0.7539\\
	-10	0.5112\\	
	-6	0.0452\\
	-2	0\\
};

\addplot[ybar, bar width=0.356, fill=mycolor4, draw=black, area legend] table[row sep=crcr] {%
	-18	0.629796297990205\\
	-14	0.546844614495926\\
	-10	0.376753073876106\\
	-6	0.197797679775553\\
	-2	0.0771908083186376\\
};

\coordinate (c2) at (rel axis cs:0.5,0);

\nextgroupplot[title={(c) Third priority class}, align=left,
major x tick style = transparent,
xlabel={$\theta$ [dB]},
ybar=4*\pgflinewidth,
bar width=9pt,
ymajorgrids = true,
ylabel={TSP},
scaled y ticks = false,
ymin=0, 
ymax=1,
ytick={0,0.2,0.4,0.6,0.8,1}]

\addplot[ybar, bar width=0.356, fill=mycolor1, draw=black, area legend] table[row sep=crcr] {%
	-18	0.5791\\
	-14	0.5143\\
	-10	0.4223\\
	-6	0.2326\\
	-2	0.0307\\
};
\addplot[ybar, bar width=0.356, fill=mycolor2, draw=black, area legend] table[row sep=crcr] {%
	-18 0.56969\\
	-14	0.4454\\
	-10	0.2708\\
	-6	0.01406\\
	-2	0\\
};
\addplot[ybar, bar width=0.356, fill=mycolor3, draw=black, area legend] table[row sep=crcr] {%
	-18	0.5783\\
	-14	0.4507\\
	-10	0.2202\\	
	-6  0.00063\\
	-2	0\\
};

\addplot[ybar, bar width=0.356, fill=mycolor4, draw=black, area legend] table[row sep=crcr] {%
	-18	0.688979930973781\\
	-14	0.613850977475116\\
	-10	0.458921623893824\\	
	-6  0.26642354973061\\
	-2	0.109568406578523\\
};

\end{groupplot}

\coordinate (c3) at ($(c1)!1.05!(c2)$);
\node[below] at (c3 |- current bounding box.south)
{\ref{grouplegend1}};
\end{tikzpicture}
	\fi	
	\caption{Comparison of different allocation strategies for three priority classes.}
	\label{fig:TSP_bar} 
\end{figure}
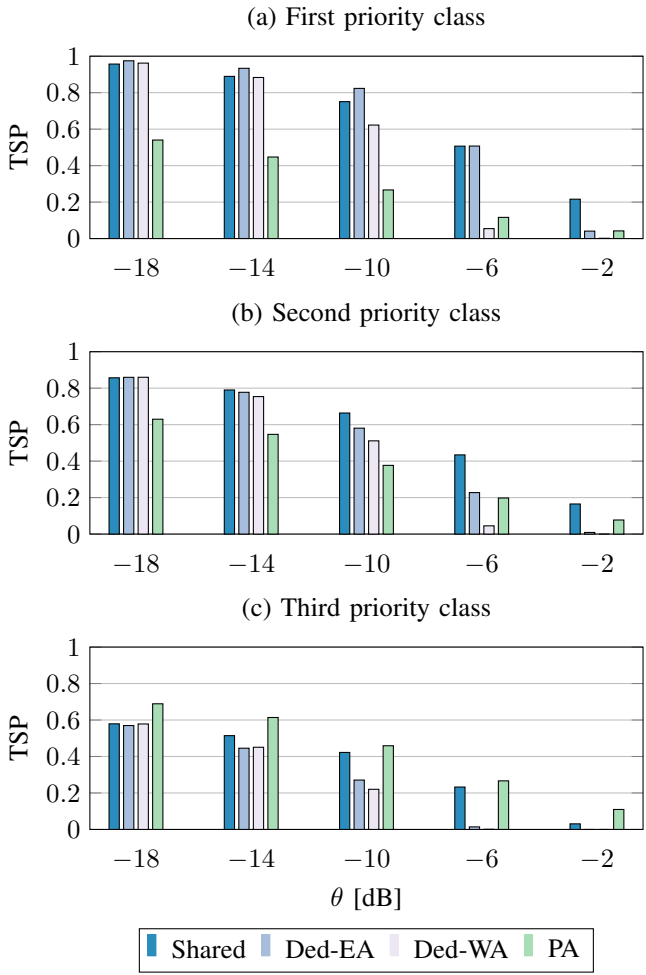

To better assess the performance of the different allocation strategies, Fig. \ref{fig:TSP_bar} compares the considered strategies against the \ac{PA} strategy. First, it is observed that for lower values of $\theta$, the dedicated-EA strategy outperforms the shared strategy. This is attributed to the successful packets transmission attempts from the first priority class while benefiting from interference protection from lower priority classes. As $\theta$ increases, the shared strategy outperforms the dedicated one, which results from the head of queue effect of the higher priority packets. In the dedicated strategy, when several devices have high priority packets, they keep interfering on a subset of the available channels leaving other channels for lower priority packets underutilized. Moreover, the dedicated-WA strategy fails to provide gains in the high $\theta$ region, due to the strong interference experienced by the higher priority packets, that are allocated a smaller number of channels (i.e., compared to the dedicated-EA strategy). Thus, hindering the transmission of lower priority packets, that are assigned larger pool of channels, due to the imposed priority-aware transmission discipline. Additionally, the shared channel allocation strategy alleviates the additional  overhead required for channel allocation procedures, that is essential for the dedicated strategies. For the \ac{PA} scheme, we observe the performance deterioration experienced by the higher priority classes (e.g. first and second classes), which results from the \ac{PA} negligence of higher priority traffic. For the \ac{PA} scheme, a given packet is granted service depending on its arrival time, not its priority. Accordingly, depending on the arrival probability, transmission probability is larger for traffic with higher arrival probabilities (i.e., third class has larger transmission probability compared to second and first classes). For the third priority class, due to the \ac{FCFS} nature of the \ac{PA} scheme, it outperforms the priority-aware strategies. Accordingly, the \ac{TSP} values depict a flipped behavior among the higher and lower priority classes.  

\begin{figure} 
	\centering
	\ifCLASSOPTIONdraftcls
	% This file was created by matlab2tikz.
%
%The latest updates can be retrieved from
%  http://www.mathworks.com/matlabcentral/fileexchange/22022-matlab2tikz-matlab2tikz
%where you can also make suggestions and rate matlab2tikz.
%
\definecolor{mycolor1}{rgb}{0.4, 1.0, 0.0}
\definecolor{mycolor2}{rgb}{0.03, 0.91, 0.87}
\definecolor{mycolor3}{rgb}{1,0.0,1.0}%%
\definecolor{mycolor4}{rgb}{0.69, 0.4, 0.0}%
\definecolor{mycolor5}{rgb}{0.97690,0.98390,0.08050}%

\begin{tikzpicture}[scale=0.8]
\begin{axis}[%
width=1.1*\columnwidth,
height=4cm,
scale only axis,
bar shift auto,
xmin=-19.8488888888889,
xmax=-14.1,
xtick={-19, -17, -15},
xlabel style={font=\color{white!15!black}},
xlabel={$\theta$ [dB]},
ymin=0,
ymax=10,
ylabel style={font=\color{white!15!black}},
ylabel={Mean delay [time slots]},
axis background/.style={fill=white},
xmajorgrids,
ymajorgrids,
legend style={ legend columns=3, at={(0.01,0.58)}, anchor=south west, legend cell align=left, align=left, draw=white!15!black}
]
\addplot[ybar, bar width=0.15, fill=red,line width=1.0pt, draw=black, area legend] table[row sep=crcr] {%
	-19	1.94741119855404\\
	-17	2.16507623145837\\
	-15	2.25213605301714\\
};
\addlegendentry{PA - i = 1}

\addplot[ybar, bar width=0.15, fill=red,line width=1.0pt, dashed, draw=black, area legend] table[row sep=crcr] {%
	-19	1.94741119855404\\
	-17	2.16507623145837\\
	-15	2.25213605301714\\
};
\addlegendentry{PA - i = 2}

\addplot[ybar, bar width=0.15, fill=red,line width=1.0pt, dotted, draw=black, area legend] table[row sep=crcr] {%
	-19	1.94741119855404\\
	-17	2.16507623145837\\
	-15	2.25213605301714\\
};
\addlegendentry{PA - i = 3}

\addplot[ybar, bar width=0.15, fill=mycolor1,line width=1.0pt, draw=black, area legend] table[row sep=crcr] {%
	-19	1.01911429432154\\
	-17	1.02905551126438\\
	-15	1.05492593778342\\
};
\addlegendentry{Shared - i = 1}

\addplot[ybar, bar width=0.15, fill=mycolor1,line width=1.0pt, dashed, draw=black, area legend] table[row sep=crcr] {%
	-19	1.18396503980934\\
	-17	1.20615672988849\\
	-15	1.26043706449496\\
};
\addlegendentry{Shared - i = 2}

\addplot[ybar, bar width=0.15, fill=mycolor1,line width=1.0pt, dotted, draw=black, area legend] table[row sep=crcr] {%
	-19	2.3940279688056\\
	-17	2.80512050610173\\
	-15	3.58637273481812\\
};
\addlegendentry{Shared - i = 3}

\addplot[ybar, bar width=0.15, fill=mycolor2,line width=1.0pt, draw=black, area legend] table[row sep=crcr] {%
	-19	1.02855854103283\\
	-17	1.03718349529866\\
	-15	1.06462820282271\\
};
\addlegendentry{Ded-EA - i = 1}

\addplot[ybar, bar width=0.15, fill=mycolor2,line width=1.0pt, dashed, draw=black, area legend] table[row sep=crcr] {%
	-19	1.23818275145404\\
	-17	1.25374125368888\\
	-15	1.34975907246833\\
};
\addlegendentry{Ded-EA - i = 2}

\addplot[ybar, bar width=0.15, fill=mycolor2,line width=1.0pt, dotted, draw=black, area legend] table[row sep=crcr] {%
	-19	3.5944947537204\\
	-17	3.64556083130286\\
	-15	7.99349490003494\\
};
\addlegendentry{Ded-EA - i = 3}

\end{axis}
\end{tikzpicture}%
	\vspace{-0.15in}
	\else
	% This file was created by matlab2tikz.
%
%The latest updates can be retrieved from
%  http://www.mathworks.com/matlabcentral/fileexchange/22022-matlab2tikz-matlab2tikz
%where you can also make suggestions and rate matlab2tikz.

\definecolor{mycolor1}{RGB}{43, 140, 190}
\definecolor{mycolor2}{RGB}{166, 189, 219}
\definecolor{mycolor4}{RGB}{168, 221, 181}

\begin{tikzpicture}[scale=0.8]
\begin{axis}[%
width=1.1*\columnwidth,
height=4cm,
scale only axis,
bar shift auto,
xmin=-19.8488888888889,
xmax=-14.1,
xtick={-19, -17, -15},
xlabel style={font=\color{white!15!black}},
xlabel={$\theta$ [dB]},
ymin=0,
ymax=10,
ylabel style={font=\color{white!15!black}},
ylabel={Average packet delay [time slots]},
axis background/.style={fill=white},
xmajorgrids,
ymajorgrids,
legend style={ legend columns=3, at={(0.005,0.62)}, anchor=south west, legend cell align=left, align=left, draw=white!15!black}
]
\addplot[ybar, bar width=0.15, fill=mycolor4,line width=1.0pt, area legend] table[row sep=crcr] {%
	-19	1.94741119855404\\
	-17	2.16507623145837\\
	-15	2.25213605301714\\
};
\addlegendentry{PA - i = 1}

\addplot[ybar, bar width=0.15, fill=mycolor4,line width=1.0pt,postaction={
	pattern=north east lines
},  area legend] table[row sep=crcr] {%
	-19	1.94741119855404\\
	-17	2.16507623145837\\
	-15	2.25213605301714\\
};
\addlegendentry{PA - i = 2}

\addplot[ybar, bar width=0.15, fill=mycolor4,line width=1.0pt, postaction={
	pattern=crosshatch
}, area legend] table[row sep=crcr] {%
	-19	1.94741119855404\\
	-17	2.16507623145837\\
	-15	2.25213605301714\\
};
\addlegendentry{PA - i = 3}

\addplot[ybar, bar width=0.15, fill=mycolor1,line width=1.0pt, area legend] table[row sep=crcr] {%
	-19	1.01911429432154\\
	-17	1.02905551126438\\
	-15	1.05492593778342\\
};
\addlegendentry{Shared - i = 1}

\addplot[ybar, bar width=0.15, fill=mycolor1,line width=1.0pt,postaction={
	pattern=north east lines
}, area legend] table[row sep=crcr] {%
	-19	1.18396503980934\\
	-17	1.20615672988849\\
	-15	1.26043706449496\\
};
\addlegendentry{Shared - i = 2}

\addplot[ybar, bar width=0.15, fill=mycolor1,line width=1.0pt, postaction={
	pattern=crosshatch
}, draw=black, area legend] table[row sep=crcr] {%
	-19	2.3940279688056\\
	-17	2.80512050610173\\
	-15	3.58637273481812\\
};
\addlegendentry{Shared - i = 3}

\addplot[ybar, bar width=0.15, fill=mycolor2,line width=1.0pt,  area legend] table[row sep=crcr] {%
	-19	1.02855854103283\\
	-17	1.03718349529866\\
	-15	1.06462820282271\\
};
\addlegendentry{Ded-EA - i = 1}

\addplot[ybar, bar width=0.15, fill=mycolor2,line width=1.0pt, postaction={
	pattern=north east lines
}, draw=black, area legend] table[row sep=crcr] {%
	-19	1.23818275145404\\
	-17	1.25374125368888\\
	-15	1.34975907246833\\
};
\addlegendentry{Ded-EA - i = 2}

\addplot[ybar, bar width=0.15, fill=mycolor2,line width=1.0pt, postaction={
	pattern=crosshatch
}, area legend] table[row sep=crcr] {%
	-19	3.5944947537204\\
	-17	3.64556083130286\\
	-15	7.99349490003494\\
};
\addlegendentry{Ded-EA - i = 3}

\end{axis}
\end{tikzpicture}%
	\fi	
	\caption{Average packet delay for priority agnostic, shared, and dedicated-equal strategies.}
	\label{fig:packet_delay} 
\end{figure}
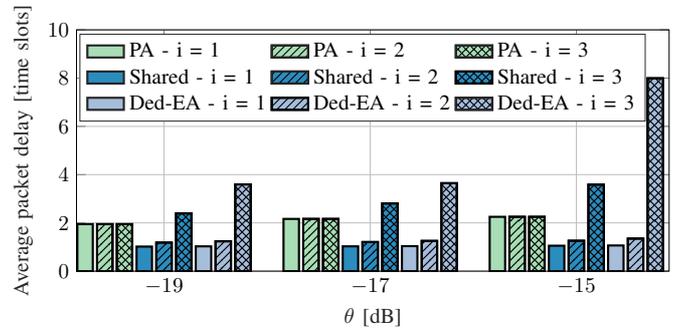

To further investigate the prioritization effect, the average packet delay is shown in Fig. \ref{fig:packet_delay}.\footnote{The delay is defined as the time elapsed from packet generation at the device until its successful reception at the BS.} Due to its priority negligence of the PA strategy, the packets belonging to the three classes experiences nearly the same waiting time with different values of $\theta$. This is attributed to the inter-class \ac{FCFS} discipline of the \ac{PA}. However, for the priority-aware strategies, high priority packets experience lower packet delays when compared to lower priority packets. The figure also highlights the traffic prioritization cost  on lower priority packets, which is due to the service interruption upon higher priority packets arrival. Hence, it is important to ensure that the prioritized transmission offers a differentiated service that meets the \ac{QoS} requirement for all priority classes.  

\begin{figure}
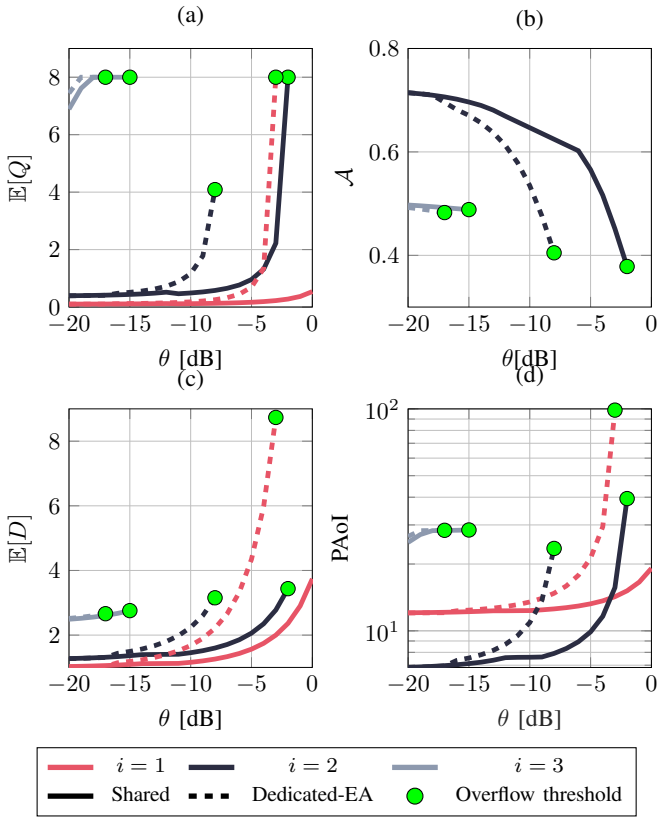

	\centering
	\ifCLASSOPTIONdraftcls
	\input{simulation_results/figures/perf_evals/1col/perf_evals.tex}
	\vspace{-0.25in}
	\else
	\input{simulation_results/figures/perf_evals/2col/perf_evals.tex}
	\fi
	\caption{Performance evaluation for shared and dedicated-equal allocation strategies (a) average number of packets (b) transmission availability (c) transmission delay (e) peak age of information.}
	\label{fig:perf_evals} 
\end{figure}

Throughout the rest of this section, we will focus on assessing the shared and dedicated-EA strategies due to their promised performance superiority as shown in Fig. To this end, \ref{fig:TSP_bar}. Fig. \ref{fig:perf_evals} showcases different \acp{KPI} under the mentioned strategies. As a common behavior in all the sub-figures,  we observe a large performance superiority of the shared allocation strategy over the dedicated-EA one in the high $\theta$ regime. As $\theta$ increases, packets transmission is subjected to a more stringent requirement on the achieved SINR. This leads to increased retransmissions, thus, increasing the aggregate network interference. Furthermore, it can be interpreted that for the low $\theta$ regime, head of the queue is determined by the arrival priority, whereas for the high $\theta$ regime, head of the queue is determined by the prioritized-based preemption discipline. In details, Fig. \ref{fig:perf_evals}(a) presents the average number of packets, where it is observed that the shared strategy results in lower number of packets residing in the queues at the high $\theta$ regime. Within a given channel allocation strategy, as the priority of the queue gets lower, its average number of packets increases. Packets residing in a given queue will have to wait until all the higher queues are served, while new packets might arrive and accumulate in the queues. The figure also highlights the effect of the queue's priority on the overflow threshold. The transmission availability is presented in Fig. \ref{fig:perf_evals}(b). For the first priority class, such a metric equals one as highest priority packets will be served upon their arrival. However, for lower priority packets, the transmission availability decreases. Fig .\ref{fig:perf_evals}(c) demonstrates the transmission delay, where it can be observed the superiority of the shared over the dedicated strategy. Finally, Fig .\ref{fig:perf_evals}(d) shows the \ac{PAoI}. We observe a flipped behavior between the first and second priority classes when considering a given allocation strategy. This is justified based on the PAoI sensitivity to the inter-arrival delays (recall $\alpha_1 = 0.1$ and $\alpha_2=0.25$). Such a behavior is expected, since \ac{PAoI} is lower when packets with low queueing delays are delivered regularly. Thus, larger inter-arrival times increases the PAoI. As $\theta$ increases, the queueing delays start to dominate the \ac{PAoI}, yielding the queues eventually in an overflow state. Finally, via observing the reported results in Figure \ref{fig:perf_evals}, it can be concluded that the exclusive resource partitioning for prioritized grant-free uplink traffic in \ac{IoT} systems is outperformed by the shared channel allocation strategy.

In addition, Fig. \ref{fig:PMF_WT} presents the average queueing delay distribution over the first five time slots. The queueing delay distributions is dependent on the prioritization and the allocation strategy. In specific, the distribution tail decays for higher priority classes, whereas for the lower classes, it takes longer to dispatch their packets. We observe also a larger tail for the dedicated strategy, when compared to the shared one over the considered priority classes.

\begin{figure} 
	\centering
	\definecolor{mycolor1}{RGB}{43,45,66}%
\definecolor{mycolor2}{RGB}{141,153,174}%
\definecolor{mycolor3}{RGB}{231,84,102}%()

\begin{tikzpicture}
\begin{axis}[%
height=1.8in,
width=0.8\columnwidth,
scale only axis,
bar shift auto,
log origin=infty,
xmin=0.506666666666667,
xmax=5.49333333333333,
xtick={1, 2, 3, 4, 5},
xlabel={$n$ [Time slot]},
ymode=log,
ymin=1e-04,
ymax=1,
ylabel={$\mathbb{P}\{W=n\}$},
legend style={legend columns=3,fill=none,draw=black,anchor=center,align=left, column sep=0.13cm, at={(0.6,0.92)}}
]

\addplot[ybar, bar width=0.1, fill=mycolor3, draw=black, area legend] table[row sep=crcr] {%
	1	0.884179410263595\\
	2	0.0982421566959545\\
	3	0.0149105023308881\\
	4	0.00226301098465761\\
	5	0.000343463862117487\\
};
\addlegendentry{\footnotesize $i=1$}]

\addplot[ybar, bar width=0.1, fill=mycolor1, draw=black, area legend] table[row sep=crcr] {%
	1	0.613549088960524\\
	2	0.204516362986842\\
	3	0.0784372902648972\\
	4	0.0468422508326947\\
	5	0.0251669102589455\\
};
\addlegendentry{\footnotesize $i=2$}]

\addplot[ybar, bar width=0.1, fill=mycolor2, draw=black, area legend] table[row sep=crcr] {%
	1	0.066800306872735\\
	2	0.0359693960083956\\
	3	0.0308713684813394\\
	4	0.0317908219896159\\
	5	0.0317002025595595\\
};
\addlegendentry{\footnotesize $i=3$}]

\addplot[ybar, bar width=0.1, fill=mycolor3, postaction={
	pattern=crosshatch
}, area legend] table[row sep=crcr] {%
	1	0.832285593550967\\
	2	0.0924761770612179\\
	3	0.041485665893205\\
	4	0.0186108523221424\\
	5	0.00834899999070087\\
};

\addplot[ybar, bar width=0.1, fill=mycolor1, postaction={
	pattern=crosshatch
}, area legend] table[row sep=crcr] {%
	1	0.410661764258016\\
	2	0.136887254752672\\
	3	0.0949320714419968\\
	4	0.072463793512572\\
	5	0.0569194457051267\\
};

\addplot[ybar, bar width=0.1, fill=mycolor2, postaction={
	pattern=crosshatch
}, area legend] table[row sep=crcr] {%
	1	0.00242661217712157\\
	2	0.00140663732614359\\
	3	0.00142595831395413\\
	4	0.00165669166852949\\
	5	0.0019374494474004\\
};
\end{axis}
\end{tikzpicture}%
	\caption{Waiting time distribution with $\theta = -10$ dB for shared (dedicated-EA) represented by solid (hashed) bars.}
	\label{fig:PMF_WT} 
\end{figure}
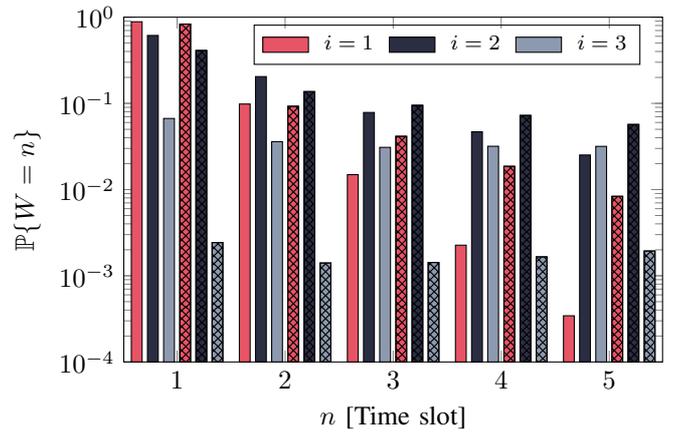
 
In Fig. \ref{fig:kappa_effect}, we investigate the effect of network scalability and devices densification of the first two priority queues under shared and dedication-EA allocation strategies. The considered values of $\kappa$ represent a network with $640$ and 5120 $\text{ device/KM}^2$, given that $\lambda= 10 \text{ BS/KM}^2$ and $C = 64$ channels. First, focusing on the first priority class (c.f. Fig. \ref{fig:kappa_effect}(a)), we observe a slight superiority of the dedicated-EA over the shared strategy over $\theta\in[-20,-6]$ dB. As mentioned earlier, such performance superiority is attributed to the successful packets transmission attempts from the first priority class while benefiting from interference protection from lower priority classes. Such a behavior is also reflected for $\kappa = 8$ within the range  $\theta \in [-20,-14.8]$ dB. Furthermore, as  $\kappa$ increases, a given device experiences stronger interference which degrades the \ac{TSP} and shifts the overflow-region threshold to lower values of $\theta$. For the second priority class (c.f. Fig. \ref{fig:kappa_effect}(b)), we observe the superiority of the shared over the dedicated-EA strategy for the two values of $\kappa$. This is due to the fact that lower priority classes experience head of the queue problem more severely under the dedicated-EA strategy. Finally, since $\kappa$ implicitly considers the number of deployed channels at every BS, such a study can help in deriving the minimum number of channels required to meet a targeted requirement.

\begin{figure} 
	\centering
	\ifCLASSOPTIONdraftcls
	\definecolor{mycolor1}{rgb}{0.49020,0.18039,0.56078}%

\begin{tikzpicture}[scale=0.9]
\begin{groupplot}[group style={
	group name=myplot,
	group size= 2 by 1,  horizontal sep=2cm}, height=1.5in,width=2.7in]

\nextgroupplot[title={{(a)}},
scale only axis,
xmin=-20,
xmax=0,
xlabel style={font=\color{white!15!black}},
xlabel={$\theta$ [dB]},
ymin=0,
ymax=1,
ylabel={TSP},
axis background/.style={fill=white},
xmajorgrids,
ymajorgrids
]

%%%%%%%%%%%%%%%%%%%%%  kappa = 1-shared  %%%%%%%%%%%%%%%%%%%%%%%%%%%%%%%
\addplot [color=black, line width=2.0pt]
table[row sep=crcr]{%
	-20	0.970586419768648\\
	-19	0.963157541669851\\
	-18	0.953911591607296\\
	-17	0.942437506889913\\
	-16	0.928249855888286\\
	-15	0.910785853735535\\
	-14	0.889408922520346\\
	-13	0.863422946040482\\
	-12	0.832102530739579\\
	-11	0.794745304002353\\
	-10	0.750751786583538\\
	-9	0.699735545235645\\
	-8	0.641659920211728\\
	-7	0.576986927269488\\
	-6	0.506809956575843\\
	-5	0.432928606613245\\
	-4	0.357818591125136\\
	-3	0.284460253687674\\
	-2	0.216019683166065\\
	-1	0.155421099123583\\
	0	0.104892819535421\\
};
%%%%%%%%%%%%%%%%%%%%%  kappa = 1-dedicated  %%%%%%%%%%%%%%%%%%%%%%%%%%%%%%%
\addplot  [color=mycolor1, line width=2.0pt]
table[row sep=crcr]{%
	-20	0.983838261520072\\
	-19	0.979969152365725\\
	-18	0.974817470518088\\
	-17	0.968350993309098\\
	-16	0.958781616556208\\
	-15	0.947769475778786\\
	-14	0.933705110560617\\
	-13	0.915679918041579\\
	-12	0.892482381072275\\
	-11	0.862485494230015\\
	-10	0.823497086433804\\
	-9	0.772582008495268\\
	-8	0.705919895220539\\
	-7	0.61892784165829\\
	-6	0.507309453272374\\
	-5	0.370587415733334\\
	-4	0.220753097400419\\
	-3	0.096710675205912\\
};

\addplot [color=black, only marks, mark size=3 pt, mark=*, mark options={solid, fill=green}]
table[row sep=crcr]{%
	-3	0.096710675205912\\
};

\addplot [color=black, only marks, mark size=3 pt, mark=*, mark options={solid, fill=green}]
table[row sep=crcr]{%
	-2	0.0409234152756813\\
};

\addplot [color=mycolor1, line width=2.0pt]
table[row sep=crcr]{%
	-2	0.0409234152756813\\
	-1	0.0207200653899588\\
	0	0.00986923635103234\\
};

%%%%%%%%%%%%%%%% kappa = 8 shared %%%%%%%%%%%%%%

\addplot [color=black, dotted, line width=2.0pt, mark=square, mark options={solid, black}]
table[row sep=crcr]{%
	-20	0.845252784449903\\
	-19	0.809731891697623\\
	-18	0.767388589736471\\
	-17	0.717593764425436\\
	-16	0.660012224039258\\
	-15	0.594791105665185\\
	-14	0.522768757682629\\
	-13	0.44565903664679\\
	-12	0.36613681412571\\
	-11	0.287731899354133\\
	-10	0.214454713930227\\
	-9	0.150151526976754\\
};

\addplot [color=black, only marks, mark size=3 pt, mark=*, mark options={solid, fill=green}]
table[row sep=crcr]{%
	-9	0.150151526976754\\
};

\addplot [color=black, only marks, mark size=3 pt, mark=*, mark options={solid, fill=green}]
table[row sep=crcr]{%
	-8	0.0977180099270141\\
};

\addplot [color=black, dotted, line width=2.0pt, mark=square, mark options={solid, black}]
table[row sep=crcr]{%
	-8	0.0977180099270141\\
	-7	0.0584328832899488\\
	-6	0.031707562510077\\
	-5	0.015405013156839\\
	-4	0.00660468747371049\\
	-3	0.00245942771830345\\
	-2	0.000781477334106164\\
	-1	0.000207659232590731\\
	0	4.50823097079544e-05\\	
};

%%%%%%%%%%%%% kappa = 8 dedicated %%%%%%%%%%%%%

\addplot [color=mycolor1, dotted, line width=2.0pt, mark=square, mark options={solid, mycolor1}]
table[row sep=crcr]{%
	-20	0.924564934957761\\
	-19	0.898694805634154\\
	-18	0.862040195879601\\
	-17	0.809250414573718\\
	-16	0.732410196545309\\
	-15	0.621088600802357\\
	-14	0.466224658956815\\
	-13	0.275555604239976\\
	-12	0.106295408085909\\
};

\addplot [color=black, only marks, mark size=3 pt, mark=*, mark options={solid, fill=green}]
table[row sep=crcr]{%
	-12	0.106295408085909\\
};

\addplot [color=black, only marks, mark size=3 pt, mark=*, mark options={solid, fill=green}]
table[row sep=crcr]{%
	-11	0.035915120588961\\
};

\addplot [color=mycolor1, dotted, line width=2.0pt, mark=square, mark options={solid, mycolor1}]
table[row sep=crcr]{%
	-11	0.035915120588961\\
	-10	0.0158620477108088\\
	-9	0.00638128877838559\\
	-8	0.00215904658432682\\
	-7	0.00059729672939644\\
	-6	0.000131012619397548\\
	-5	2.20005375317852e-05\\
	-4	2.71823619352388e-06\\
	-3	2.36225203008053e-07\\
	-2	1.37256762952001e-08\\
	-1	5.03808120066551e-10\\
	0	1.09665938119191e-11\\
};

\nextgroupplot[title={{\smash{(b)}}},
scale only axis,
xmin=-20,
xmax=0,
xlabel style={font=\color{white!15!black}},
xlabel={$\theta$ [dB]},
ymin=0,
ymax=1,
ylabel={TSP},
axis background/.style={fill=white},
xmajorgrids,
ymajorgrids,
legend style={legend columns=5,column sep=0.13cm, legend cell align=left, at={(-0.2,-0.35)},anchor=north},
]

\addplot [color=black, line width=2.0pt]
table[row sep=crcr]{%
	-20	0.872688906555566\\
	-19	0.865673839867\\
	-18	0.856883098367325\\
	-17	0.845881572593045\\
	-16	0.832143998435199\\
	-15	0.815045744104584\\
	-14	0.790108258444602\\
	-13	0.76421882236657\\
	-12	0.732981292300075\\
	-11	0.705393846916411\\
	-10	0.663799271915771\\
	-9	0.615627319384566\\
	-8	0.560880862534403\\
	-7	0.500046702209288\\
	-6	0.434221209186397\\
	-5	0.365178822552751\\
	-4	0.295337991039018\\
	-3	0.227589692709488\\
	-2	0.164983421958394\\
};\addlegendentry{\small Shared-$\kappa = 1$}

\addplot [color=mycolor1, line width=2.0pt]
table[row sep=crcr]{%
	-20	0.874009840974099\\
	-19	0.867834441148821\\
	-18	0.859517940874971\\
	-17	0.849048738246672\\
	-16	0.823751014822288\\
	-15	0.803377103650015\\
	-14	0.777416122700506\\
	-13	0.744304384418031\\
	-12	0.702082803316479\\
	-11	0.6484010877248\\
	-10	0.580743098940922\\
	-9	0.497398737910211\\
	-8	0.400780631556914\\
};\addlegendentry{\small Dedicated-EA-$\kappa = 1$}

\addplot [color=black, dotted, line width=2.0pt, mark=square, mark options={solid, black}]
table[row sep=crcr]{%
	-20	0.746754802617697\\
	-19	0.711475033258048\\
	-18	0.669363800237468\\
	-17	0.619762154783212\\
	-16	0.585325276091175\\
	-15	0.525338033862888\\
	-14	0.459286691480197\\
	-13	0.388834941594406\\
	-12	0.316535525231341\\
	-11	0.245714162933435\\
	-10	0.180098719057133\\
	-9	0.123197472013976\\
	-8	0.0775567072453511\\
	-7	0.044152596785429\\
};\addlegendentry{\small Shared-$\kappa = 8$}

\addplot [color=mycolor1, dotted, line width=2.0pt, mark=square, mark options={solid, mycolor1}]
table[row sep=crcr]{%
	-20	0.701457986263273\\
	-19	0.641246451153093\\
	-18	0.562749300356252\\
	-17	0.464330349923018\\
	-16	0.354593375085902\\
};\addlegendentry{\small Dedicated-EA-$\kappa = 8$}

\addplot [color=black, only marks, mark size=3 pt, mark=*, mark options={solid, fill=green}]
table[row sep=crcr]{%
	-7	0.304220008946337\\
};\addlegendentry{\small Overflow threshold}

\addplot [color=mycolor1, line width=2.0pt]
table[row sep=crcr]{%
	-7	0.304220008946337\\
	-6	0.227498515751832\\
	-5	0.172251711443045\\
	-4	0.126831485410998\\
	-3	0.0700510223387054\\
	-2	0.00933417669547961\\
	-1	0.000196890349453315\\
	0	1.35761304708741e-06\\
};

\addplot [color=black, draw=none,mark size=3 pt, mark=*, mark options={solid, fill=green}, forget plot]
table[row sep=crcr]{%
	-8	0.400780631556914\\
};

\addplot [color=black, draw=none,mark size=3 pt, mark=*, mark options={solid, fill=green}, forget plot]
table[row sep=crcr]{%
	-7	0.044152596785429\\
};

\addplot [color=black, dotted, line width=2.0pt, mark=square, mark options={solid, black}]
table[row sep=crcr]{%
	-6	0.0221976251420623\\
	-5	0.00949534503130587\\
	-4	0.0032077849570385\\
	-3	0.000692971728400592\\
	-2	4.94708689525438e-05\\
	-1	5.85921538059682e-07\\
	0	1.66884179387225e-09\\
};

\addplot [color=black, draw=none,mark size=3 pt, mark=*, mark options={solid, fill=green}, forget plot]
table[row sep=crcr]{%
	-6	0.0221976251420623\\
};

\addplot [color=black, line width=2.0pt]
table[row sep=crcr]{%
	-1	0.110308981686767\\
	0	0.0656557082804053\\
};

\addplot [color=black, draw=none,mark size=3 pt, mark=*, mark options={solid, fill=green}, forget plot]
table[row sep=crcr]{%
	-2	0.164983421958394\\
};

\addplot [color=black, draw=none,mark size=3 pt, mark=*, mark options={solid, fill=green}, forget plot]
table[row sep=crcr]{%
	-1	0.110308981686767\\
};

\addplot [color=mycolor1, dotted, line width=2.0pt, mark=square, mark options={solid, mycolor1}]
table[row sep=crcr]{%
	-15	0.260053274725428\\
	-14	0.197245567173026\\
	-13	0.155092358644666\\
	-12	0.0706383377515074\\
	-11	0.0165140434428533\\
	-10	0.000241865242614196\\
	-9	8.90700423746265e-07\\
	-8	9.83557565099087e-10\\
	-7	2.9044415877958e-13\\
	-6	4.42144517000077e-16\\
	-5	-6.47698417085242e-16\\
	-4	5.1607533329374e-16\\
	-3	-5.05201041933887e-16\\
	-2	5.05981873342391e-16\\
	-1	7.61768277210408e-16\\
	0	-1.03576237439907e-16\\
};

\addplot [color=black, draw=none,mark size=3 pt, mark=*, mark options={solid, fill=green}, forget plot]
table[row sep=crcr]{%
	-15	0.260053274725428\\
};

\addplot [color=black, draw=none,mark size=3 pt, mark=*, mark options={solid, fill=green}, forget plot]
table[row sep=crcr]{%
	-16	0.354593375085902\\
};

\end{groupplot}

\end{tikzpicture}
	\vspace{-0.05in}
	\else
	\definecolor{mycolor1}{RGB}{43,45,66}%
\definecolor{mycolor2}{RGB}{127,201,127}%

\begin{tikzpicture}[scale=0.9]
\begin{groupplot}[group style={
	group name=myplot,
	group size= 1 by 2,  vertical sep=2cm}, height=1.5in,width=0.93*\columnwidth]

\nextgroupplot[title={{(a)}},
scale only axis,
xmin=-20,
xmax=0,
xlabel style={font=\color{white!15!black}},
xlabel={$\theta$ [dB]},
ymin=0,
ymax=1,
ylabel={TSP},
axis background/.style={fill=white},
xmajorgrids,
ymajorgrids
]

%%%%%%%%%%%%%%%%%%%%%  kappa = 1-shared  %%%%%%%%%%%%%%%%%%%%%%%%%%%%%%%
\addplot [color=mycolor2, line width=2.0pt]
table[row sep=crcr]{%
	-20	0.970586419768648\\
	-19	0.963157541669851\\
	-18	0.953911591607296\\
	-17	0.942437506889913\\
	-16	0.928249855888286\\
	-15	0.910785853735535\\
	-14	0.889408922520346\\
	-13	0.863422946040482\\
	-12	0.832102530739579\\
	-11	0.794745304002353\\
	-10	0.750751786583538\\
	-9	0.699735545235645\\
	-8	0.641659920211728\\
	-7	0.576986927269488\\
	-6	0.506809956575843\\
	-5	0.432928606613245\\
	-4	0.357818591125136\\
	-3	0.284460253687674\\
	-2	0.216019683166065\\
	-1	0.155421099123583\\
	0	0.104892819535421\\
};
%%%%%%%%%%%%%%%%%%%%%  kappa = 1-dedicated  %%%%%%%%%%%%%%%%%%%%%%%%%%%%%%%
\addplot  [color=mycolor1, line width=2.0pt]
table[row sep=crcr]{%
	-20	0.983838261520072\\
	-19	0.979969152365725\\
	-18	0.974817470518088\\
	-17	0.968350993309098\\
	-16	0.958781616556208\\
	-15	0.947769475778786\\
	-14	0.933705110560617\\
	-13	0.915679918041579\\
	-12	0.892482381072275\\
	-11	0.862485494230015\\
	-10	0.823497086433804\\
	-9	0.772582008495268\\
	-8	0.705919895220539\\
	-7	0.61892784165829\\
	-6	0.507309453272374\\
	-5	0.370587415733334\\
	-4	0.220753097400419\\
	-3	0.096710675205912\\
};

\addplot [color=black, only marks, mark size=3 pt, mark=*, mark options={solid, fill=green}]
table[row sep=crcr]{%
	-3	0.096710675205912\\
};

\addplot [color=black, only marks, mark size=3 pt, mark=*, mark options={solid, fill=green}]
table[row sep=crcr]{%
	-2	0.0409234152756813\\
};

\addplot [color=mycolor1, line width=2.0pt]
table[row sep=crcr]{%
	-2	0.0409234152756813\\
	-1	0.0207200653899588\\
	0	0.00986923635103234\\
};

%%%%%%%%%%%%%%%% kappa = 8 shared %%%%%%%%%%%%%%

\addplot [color=mycolor2, dotted, line width=2.0pt, mark=square, mark options={solid, mycolor2}]
table[row sep=crcr]{%
	-20	0.845252784449903\\
	-19	0.809731891697623\\
	-18	0.767388589736471\\
	-17	0.717593764425436\\
	-16	0.660012224039258\\
	-15	0.594791105665185\\
	-14	0.522768757682629\\
	-13	0.44565903664679\\
	-12	0.36613681412571\\
	-11	0.287731899354133\\
	-10	0.214454713930227\\
	-9	0.150151526976754\\
};

\addplot [color=black, only marks, mark size=3 pt, mark=*, mark options={solid, fill=green}]
table[row sep=crcr]{%
	-9	0.150151526976754\\
};

\addplot [color=black, only marks, mark size=3 pt, mark=*, mark options={solid, fill=green}]
table[row sep=crcr]{%
	-8	0.0977180099270141\\
};

\addplot [color=mycolor2, dotted, line width=2.0pt, mark=square, mark options={solid, mycolor2}]
table[row sep=crcr]{%
	-8	0.0977180099270141\\
	-7	0.0584328832899488\\
	-6	0.031707562510077\\
	-5	0.015405013156839\\
	-4	0.00660468747371049\\
	-3	0.00245942771830345\\
	-2	0.000781477334106164\\
	-1	0.000207659232590731\\
	0	4.50823097079544e-05\\	
};

%%%%%%%%%%%%% kappa = 8 dedicated %%%%%%%%%%%%%

\addplot [color=mycolor1, dotted, line width=2.0pt, mark=square, mark options={solid, mycolor1}]
table[row sep=crcr]{%
	-20	0.924564934957761\\
	-19	0.898694805634154\\
	-18	0.862040195879601\\
	-17	0.809250414573718\\
	-16	0.732410196545309\\
	-15	0.621088600802357\\
	-14	0.466224658956815\\
	-13	0.275555604239976\\
	-12	0.106295408085909\\
};

\addplot [color=black, only marks, mark size=3 pt, mark=*, mark options={solid, fill=green}]
table[row sep=crcr]{%
	-12	0.106295408085909\\
};

\addplot [color=black, only marks, mark size=3 pt, mark=*, mark options={solid, fill=green}]
table[row sep=crcr]{%
	-11	0.035915120588961\\
};

\addplot [color=mycolor1, dotted, line width=2.0pt, mark=square, mark options={solid, mycolor1}]
table[row sep=crcr]{%
	-11	0.035915120588961\\
	-10	0.0158620477108088\\
	-9	0.00638128877838559\\
	-8	0.00215904658432682\\
	-7	0.00059729672939644\\
	-6	0.000131012619397548\\
	-5	2.20005375317852e-05\\
	-4	2.71823619352388e-06\\
	-3	2.36225203008053e-07\\
	-2	1.37256762952001e-08\\
	-1	5.03808120066551e-10\\
	0	1.09665938119191e-11\\
};

\nextgroupplot[title={{\smash{(b)}}},
scale only axis,
xmin=-20,
xmax=0,
xlabel style={font=\color{white!15!black}},
xlabel={$\theta$ [dB]},
ymin=0,
ymax=1,
ylabel={TSP},
axis background/.style={fill=white},
xmajorgrids,
ymajorgrids,
legend style={legend columns=2,column sep=0.13cm, legend cell align=left, at={(0,-0.45)},anchor=west},
]

\addplot [color=mycolor2, line width=2.0pt]
table[row sep=crcr]{%
	-20	0.872688906555566\\
	-19	0.865673839867\\
	-18	0.856883098367325\\
	-17	0.845881572593045\\
	-16	0.832143998435199\\
	-15	0.815045744104584\\
	-14	0.790108258444602\\
	-13	0.76421882236657\\
	-12	0.732981292300075\\
	-11	0.705393846916411\\
	-10	0.663799271915771\\
	-9	0.615627319384566\\
	-8	0.560880862534403\\
	-7	0.500046702209288\\
	-6	0.434221209186397\\
	-5	0.365178822552751\\
	-4	0.295337991039018\\
	-3	0.227589692709488\\
	-2	0.164983421958394\\
};\addlegendentry{\small Shared-$\kappa = 1$}

\addplot [color=mycolor1, line width=2.0pt]
table[row sep=crcr]{%
	-20	0.874009840974099\\
	-19	0.867834441148821\\
	-18	0.859517940874971\\
	-17	0.849048738246672\\
	-16	0.823751014822288\\
	-15	0.803377103650015\\
	-14	0.777416122700506\\
	-13	0.744304384418031\\
	-12	0.702082803316479\\
	-11	0.6484010877248\\
	-10	0.580743098940922\\
	-9	0.497398737910211\\
	-8	0.400780631556914\\
};\addlegendentry{\small Dedicated-EA-$\kappa = 1$}

\addplot [color=mycolor2, dotted, line width=2.0pt, mark=square, mark options={solid, mycolor2}]
table[row sep=crcr]{%
	-20	0.746754802617697\\
	-19	0.711475033258048\\
	-18	0.669363800237468\\
	-17	0.619762154783212\\
	-16	0.585325276091175\\
	-15	0.525338033862888\\
	-14	0.459286691480197\\
	-13	0.388834941594406\\
	-12	0.316535525231341\\
	-11	0.245714162933435\\
	-10	0.180098719057133\\
	-9	0.123197472013976\\
	-8	0.0775567072453511\\
	-7	0.044152596785429\\
};\addlegendentry{\small Shared-$\kappa = 8$}

\addplot [color=mycolor1, dotted, line width=2.0pt, mark=square, mark options={solid, mycolor1}]
table[row sep=crcr]{%
	-20	0.701457986263273\\
	-19	0.641246451153093\\
	-18	0.562749300356252\\
	-17	0.464330349923018\\
	-16	0.354593375085902\\
};\addlegendentry{\small Dedicated-EA-$\kappa = 8$}

\addplot [color=black, only marks, mark size=3 pt, mark=*, mark options={solid, fill=green}]
table[row sep=crcr]{%
	-7	0.304220008946337\\
};\addlegendentry{\small Overflow threshold}

\addplot [color=mycolor1, line width=2.0pt]
table[row sep=crcr]{%
	-7	0.304220008946337\\
	-6	0.227498515751832\\
	-5	0.172251711443045\\
	-4	0.126831485410998\\
	-3	0.0700510223387054\\
	-2	0.00933417669547961\\
	-1	0.000196890349453315\\
	0	1.35761304708741e-06\\
};

\addplot [color=black, draw=none,mark size=3 pt, mark=*, mark options={solid, fill=green}, forget plot]
table[row sep=crcr]{%
	-8	0.400780631556914\\
};

\addplot [color=black, draw=none,mark size=3 pt, mark=*, mark options={solid, fill=green}, forget plot]
table[row sep=crcr]{%
	-7	0.044152596785429\\
};

\addplot [color=mycolor2, dotted, line width=2.0pt, mark=square, mark options={solid, mycolor2}]
table[row sep=crcr]{%
	-6	0.0221976251420623\\
	-5	0.00949534503130587\\
	-4	0.0032077849570385\\
	-3	0.000692971728400592\\
	-2	4.94708689525438e-05\\
	-1	5.85921538059682e-07\\
	0	1.66884179387225e-09\\
};

\addplot [color=black, draw=none,mark size=3 pt, mark=*, mark options={solid, fill=green}, forget plot]
table[row sep=crcr]{%
	-6	0.0221976251420623\\
};

\addplot [color=mycolor2, line width=2.0pt]
table[row sep=crcr]{%
	-1	0.110308981686767\\
	0	0.0656557082804053\\
};

\addplot [color=black, draw=none,mark size=3 pt, mark=*, mark options={solid, fill=green}, forget plot]
table[row sep=crcr]{%
	-2	0.164983421958394\\
};

\addplot [color=black, draw=none,mark size=3 pt, mark=*, mark options={solid, fill=green}, forget plot]
table[row sep=crcr]{%
	-1	0.110308981686767\\
};

\addplot [color=mycolor1, dotted, line width=2.0pt, mark=square, mark options={solid, mycolor1}]
table[row sep=crcr]{%
	-15	0.260053274725428\\
	-14	0.197245567173026\\
	-13	0.155092358644666\\
	-12	0.0706383377515074\\
	-11	0.0165140434428533\\
	-10	0.000241865242614196\\
	-9	8.90700423746265e-07\\
	-8	9.83557565099087e-10\\
	-7	2.9044415877958e-13\\
	-6	4.42144517000077e-16\\
	-5	-6.47698417085242e-16\\
	-4	5.1607533329374e-16\\
	-3	-5.05201041933887e-16\\
	-2	5.05981873342391e-16\\
	-1	7.61768277210408e-16\\
	0	-1.03576237439907e-16\\
};

\addplot [color=black, draw=none,mark size=3 pt, mark=*, mark options={solid, fill=green}, forget plot]
table[row sep=crcr]{%
	-15	0.260053274725428\\
};

\addplot [color=black, draw=none,mark size=3 pt, mark=*, mark options={solid, fill=green}, forget plot]
table[row sep=crcr]{%
	-16	0.354593375085902\\
};

\end{groupplot}

\end{tikzpicture}
	\fi	
	\caption{Effect of devices densification on the  (a) first (b) second priority class.}
	\label{fig:kappa_effect} 
\end{figure}
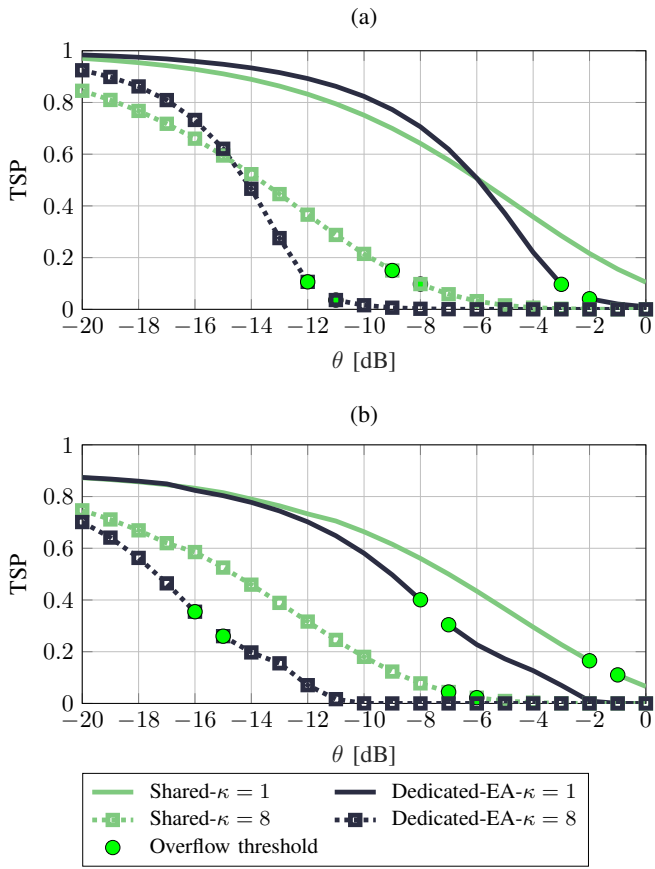

\begin{figure*}
	\centering
	\ifCLASSOPTIONdraftcls
	\definecolor{mycolor1}{rgb}{0.00000,1.00000,1.00000}%
\definecolor{mycolor2}{rgb}{0.00000,0.44700,0.74100}%
\definecolor{mycolor3}{rgb}{0.85000,0.32500,0.09800}%
\definecolor{mycolor4}{rgb}{0.49020,0.18039,0.56078}%
\definecolor{mycolor5}{rgb}{1.00000,1.00000,0.00000}%

\definecolor{blue1}{rgb}{0.0, 0.33, 0.71}
\definecolor{green1}{rgb}{0.6, 1.0, 0.6}
\definecolor{red1}{rgb}{0.9, 0.4, 0.38}
\definecolor{red2}{rgb}{0.9, 0.13, 0.13}

\begin{tikzpicture}[scale=0.9]
\begin{groupplot}[group style={
	group name=myplot,
	group size= 3 by 1, , horizontal sep=1.7cm}, height=1.5in,width=1.65in]

\nextgroupplot[title={{(a)}},
scale only axis,
xmin=0.05,
xmax=1,
xlabel={$\alpha_2$},
ymin=0,
ymax=0.8,
ylabel={$\alpha_1$},
ytick={0,0.1,0.2,0.3,0.4,0.5,0.6,0.7,0.8}, 
grid=both]

\addplot[fill=blue, draw=black] 
table[x = x, y=y]{simulation_results/figures/Pareto_regions/figa_shared_20.txt}\closedcycle;

\addplot[color=black, line width=1.5pt] 
table[x = x, y=y]{simulation_results/figures/Pareto_regions/figa_shared_20.txt};

\addplot[fill=blue1, draw=black] 
table[x = x, y=y]{simulation_results/figures/Pareto_regions/figa_dedicated_20.txt}\closedcycle;

\addplot[color=black, dotted, line width=2pt] 
table[x = x, y=y]{simulation_results/figures/Pareto_regions/figa_dedicated_20.txt};

%%%%%%%%%%%%%%%%%%%%%%%%%%%%%%%%%%%%%%%%%%%%%%%%%%%%%%%%
%%%% 			Second value of theta 				%%%%

\addplot[fill=green,  draw=black] 
table[x = x, y=y]{simulation_results/figures/Pareto_regions/figa_shared_10.txt}\closedcycle;

\addplot [color=black, line width=1.5pt]
table[x = x, y=y]{simulation_results/figures/Pareto_regions/figa_shared_10.txt};

\addplot[fill=green1, draw=black]
table[x = x, y=y]{simulation_results/figures/Pareto_regions/figa_dedicated_10.txt}\closedcycle;

\addplot [color=black, dotted, line width=2pt]
table[x = x, y=y]{simulation_results/figures/Pareto_regions/figa_dedicated_10.txt};
%%%%%%%%%%%%%%%%%%%%%%%%%%%%%%%%%%%%%%%%%%%%%%%%%%%%%%%%
%%%% 			Third value of theta 				%%%%

\addplot[fill=red2, draw=black] 
table[x = x, y=y]{simulation_results/figures/Pareto_regions/figa_shared_5.txt}\closedcycle;

\addplot[color=black, line width=1.5pt] 
table[x = x, y=y]{simulation_results/figures/Pareto_regions/figa_shared_5.txt};

\addplot[fill=red1,  draw=black] 
table[x = x, y=y]{simulation_results/figures/Pareto_regions/figa_dedicated_5.txt}\closedcycle;

\addplot [color=black, dotted, line width=2pt]
table[x = x, y=y]{simulation_results/figures/Pareto_regions/figa_dedicated_5.txt};

\node[below right, align=left]
at (rel axis cs:0.23,0.67) {\small $\theta$ = (\textcolor{blue}{-20}, \textcolor{green}{-10}, \textcolor{red}{-5}) dB};

\draw[->, line width=0.5mm](rel axis cs:0.4,0.56) -- (axis cs:0.15,0.08);

%%%%%%%%%%%%%%%%%%%%%%%%%%%%%%%%%%%%%%%%%%%%%%%%%%%%%%%%%%%%%%%%%%%%%%%%%%%%%%
\nextgroupplot[title={{\smash{(b)}}},
scale only axis,
xmin=0.01,
xmax=0.5,
xlabel={$\alpha_2$},
ymin=0.1,
ymax=0.6,
ylabel={$\alpha_1$},
grid=both,
ytick={0,0.1,0.2,0.3,0.4,0.5,0.6}, 
]

%%%%%%%%%%%% shared %%%%%%%%%%%%%%%%%%%
\addplot[fill=blue, draw=black] 
table[x = x, y=y]{simulation_results/figures/Pareto_regions/figb_shared_0.1.txt}\closedcycle;

\addplot [color=black, line width=1.5pt]
table[x = x, y=y]{simulation_results/figures/Pareto_regions/figb_shared_0.1.txt};
\addplot[fill=blue1, draw=black] 
table[x = x, y=y]{simulation_results/figures/Pareto_regions/figb_dedicated_0.1.txt}\closedcycle;

\addplot[color=black, dotted, line width=2pt]
table[x = x, y=y]{simulation_results/figures/Pareto_regions/figb_dedicated_0.1.txt};

%\addplot[fill=green, draw=black] 
%table[x = x, y=y]{simulation_results/figures/Pareto_regions/figb_shared_0.2.txt}\closedcycle;

\addplot[fill=red, draw=black] 
table[x = x, y=y]{simulation_results/figures/Pareto_regions/figb_shared_0.3.txt}\closedcycle;
\addplot [color=black, line width=1.5pt]
table[x = x, y=y]{simulation_results/figures/Pareto_regions/figb_shared_0.3.txt};

\addplot[fill=red1, draw=black] 
table[x = x, y=y]{simulation_results/figures/Pareto_regions/figb_dedicated_0.3.txt}\closedcycle;
\addplot [color=black,dotted, line width=2pt]
table[x = x, y=y]{simulation_results/figures/Pareto_regions/figb_dedicated_0.3.txt};

\node[below right, align=left]
at (rel axis cs:0.2,0.7) {\small $\alpha_3 = (\textcolor{blue}{0.1}, \textcolor{red}{0.3} )$};
\draw[->, line width=0.5mm](rel axis cs:0.43,0.56) -- (axis cs:0.05,0.12);

%%%%%%%%%%%%%%%%%%%%%%%%%%%%%%%%%%%%%%%%%%%%%%%%%%%%%%%%%%%%%%%%%%%%%%%%%%%%%%
\nextgroupplot[title={{\smash{(c)}}},
scale only axis,
xmin=-110,
xmax=-70,
xlabel={$\rho\text{ [dBm]}$},
ymin=-19,
ymax=4,
y coord trafo/.code=\pgfmathparse{#1+20},
y coord inv trafo/.code=\pgfmathparse{#1-20},
ylabel={$\theta\text{ [dB]}$},
grid=both]

%%%%%%%%%%%%%%%%%%%%%%%%%%%%%%%%%%%%%%%%%%%%%%%%%%%%%%%%
%%%% 			first value of kappa 				%%%%
\addplot[fill=blue1, draw=black] 
table[x = x, y=y]{simulation_results/figures/Pareto_regions/figc_shared_1.txt}\closedcycle;
\addplot [color=black,dotted, line width=1.5pt]
table[x = x, y=y]{simulation_results/figures/Pareto_regions/figc_shared_1.txt};

%\addplot[fill=blue1, draw=black] 
%table[x = x, y=y]{simulation_results/figures/Pareto_regions/figc_dedicated_1.txt}\closedcycle;
%\addplot [color=black,dotted,  line width=1.5pt]
%table[x = x, y=y]{simulation_results/figures/Pareto_regions/figc_dedicated_1.txt};

%%%%%%%%%%%%%%%%%%%%%%%%%%%%%%%%%%%%%%%%%%%%%%%%%%%%%%%%
%%%% 			Second value of kappa 				%%%%

\addplot[fill=green1, draw=black] 
table[x = x, y=y]{simulation_results/figures/Pareto_regions/figc_shared_4.txt}\closedcycle;
\addplot [color=black,dotted, line width=1.5pt]
table[x = x, y=y]{simulation_results/figures/Pareto_regions/figc_shared_4.txt};

%\addplot[fill=green1, draw=black] 
%table[x = x, y=y]{simulation_results/figures/Pareto_regions/figc_dedicated_4.txt}\closedcycle;
%\addplot [color=black,dotted,  line width=1.5pt]
%table[x = x, y=y]{simulation_results/figures/Pareto_regions/figc_dedicated_4.txt};

%%%%%%%%%%%%%%%%%%%%%%%%%%%%%%%%%%%%%%%%%%%%%%%%%%%%%%%%
%%%% 			third value of kappa 				%%%%

\addplot[fill=red1,  draw=black] 
table[x = x, y=y]{simulation_results/figures/Pareto_regions/figc_shared_8.txt}\closedcycle;
\addplot [color=black,dotted, line width=1.5pt]
table[x = x, y=y]{simulation_results/figures/Pareto_regions/figc_shared_8.txt};

%\addplot[fill=red1, draw=black] 
%table[x = x, y=y]{simulation_results/figures/Pareto_regions/figc_dedicated_8.txt}\closedcycle;
%\addplot [color=black,dotted,  line width=2pt]
%table[x = x, y=y]{simulation_results/figures/Pareto_regions/figc_dedicated_8.txt};

\node[below right, align=left]
at (rel axis cs:0.01,0.9) {\small $\kappa$ = (\textcolor{blue}{1}, \textcolor{green}{4}, \textcolor{red}{8})};
\draw[->, line width=0.5mm](rel axis cs:0.2,0.8) -- (rel axis cs:0.6,0.3);

\end{groupplot}

\end{tikzpicture}
	\vspace{-0.15in}
	\else
	\definecolor{mycolor1}{rgb}{0.00000,1.00000,1.00000}%
\definecolor{mycolor2}{rgb}{0.00000,0.44700,0.74100}%
\definecolor{mycolor3}{rgb}{0.85000,0.32500,0.09800}%
\definecolor{mycolor4}{rgb}{0.49020,0.18039,0.56078}%
\definecolor{mycolor5}{rgb}{1.00000,1.00000,0.00000}%

\definecolor{blue1}{rgb}{0.0, 0.33, 0.71}
\definecolor{green1}{rgb}{0.6, 1.0, 0.6}
\definecolor{red1}{rgb}{0.9, 0.4, 0.38}
\definecolor{red2}{rgb}{0.9, 0.13, 0.13}

\begin{tikzpicture}[scale=0.9]
\begin{groupplot}[group style={
	group name=myplot,
	group size= 3 by 1, , horizontal sep=1.7cm}, height=2in,width=1.9in]

\nextgroupplot[title={{(a)}},
scale only axis,
xmin=0.05,
xmax=1,
xlabel={$\alpha_2$},
ymin=0,
ymax=0.8,
ylabel={$\alpha_1$},
ytick={0,0.2,0.4,0.6,0.8}, 
grid=both]

\addplot[fill=blue, draw=black] 
table[x = x, y=y]{simulation_results/figures/Pareto_regions/figa_shared_20.txt}\closedcycle;

\addplot[color=black, line width=1.5pt] 
table[x = x, y=y]{simulation_results/figures/Pareto_regions/figa_shared_20.txt};

\addplot[fill=blue1, draw=black] 
table[x = x, y=y]{simulation_results/figures/Pareto_regions/figa_dedicated_20.txt}\closedcycle;

\addplot[color=black, dotted, line width=2pt] 
table[x = x, y=y]{simulation_results/figures/Pareto_regions/figa_dedicated_20.txt};

%%%%%%%%%%%%%%%%%%%%%%%%%%%%%%%%%%%%%%%%%%%%%%%%%%%%%%%%
%%%% 			Second value of theta 				%%%%

\addplot[fill=green,  draw=black] 
table[x = x, y=y]{simulation_results/figures/Pareto_regions/figa_shared_10.txt}\closedcycle;

\addplot [color=black, line width=1.5pt]
table[x = x, y=y]{simulation_results/figures/Pareto_regions/figa_shared_10.txt};

\addplot[fill=green1, draw=black]
table[x = x, y=y]{simulation_results/figures/Pareto_regions/figa_dedicated_10.txt}\closedcycle;

\addplot [color=black, dotted, line width=2pt]
table[x = x, y=y]{simulation_results/figures/Pareto_regions/figa_dedicated_10.txt};
%%%%%%%%%%%%%%%%%%%%%%%%%%%%%%%%%%%%%%%%%%%%%%%%%%%%%%%%
%%%% 			Third value of theta 				%%%%

\addplot[fill=red2, draw=black] 
table[x = x, y=y]{simulation_results/figures/Pareto_regions/figa_shared_5.txt}\closedcycle;

\addplot[color=black, line width=1.5pt] 
table[x = x, y=y]{simulation_results/figures/Pareto_regions/figa_shared_5.txt};

\addplot[fill=red1,  draw=black] 
table[x = x, y=y]{simulation_results/figures/Pareto_regions/figa_dedicated_5.txt}\closedcycle;

\addplot [color=black, dotted, line width=2pt]
table[x = x, y=y]{simulation_results/figures/Pareto_regions/figa_dedicated_5.txt};

\node[below right, align=left]
at (rel axis cs:0.23,0.67) {\small $\theta$ = (\textcolor{blue}{-20}, \textcolor{green}{-10}, \textcolor{red}{-5}) dB};

\draw[->, line width=0.5mm](rel axis cs:0.4,0.56) -- (axis cs:0.15,0.08);

%%%%%%%%%%%%%%%%%%%%%%%%%%%%%%%%%%%%%%%%%%%%%%%%%%%%%%%%%%%%%%%%%%%%%%%%%%%%%%
\nextgroupplot[title={{\smash{(b)}}},
scale only axis,
xmin=0.01,
xmax=0.5,
xlabel={$\alpha_2$},
ymin=0.1,
ymax=0.6,
ylabel={$\alpha_1$},
grid=both,
ytick={0,0.1,0.2,0.3,0.4,0.5,0.6}, 
]

%%%%%%%%%%%% shared %%%%%%%%%%%%%%%%%%%
\addplot[fill=blue, draw=black] 
table[x = x, y=y]{simulation_results/figures/Pareto_regions/figb_shared_0.1.txt}\closedcycle;

\addplot [color=black, line width=1.5pt]
table[x = x, y=y]{simulation_results/figures/Pareto_regions/figb_shared_0.1.txt};
\addplot[fill=blue1, draw=black] 
table[x = x, y=y]{simulation_results/figures/Pareto_regions/figb_dedicated_0.1.txt}\closedcycle;

\addplot[color=black, dotted, line width=2pt]
table[x = x, y=y]{simulation_results/figures/Pareto_regions/figb_dedicated_0.1.txt};

%\addplot[fill=green, draw=black] 
%table[x = x, y=y]{simulation_results/figures/Pareto_regions/figb_shared_0.2.txt}\closedcycle;

\addplot[fill=red, draw=black] 
table[x = x, y=y]{simulation_results/figures/Pareto_regions/figb_shared_0.3.txt}\closedcycle;
\addplot [color=black, line width=1.5pt]
table[x = x, y=y]{simulation_results/figures/Pareto_regions/figb_shared_0.3.txt};

\addplot[fill=red1, draw=black] 
table[x = x, y=y]{simulation_results/figures/Pareto_regions/figb_dedicated_0.3.txt}\closedcycle;
\addplot [color=black,dotted, line width=2pt]
table[x = x, y=y]{simulation_results/figures/Pareto_regions/figb_dedicated_0.3.txt};

\node[below right, align=left]
at (rel axis cs:0.2,0.7) {\small $\alpha_3 = (\textcolor{blue}{0.1}, \textcolor{red}{0.3} )$};
\draw[->, line width=0.5mm](rel axis cs:0.43,0.56) -- (axis cs:0.05,0.12);

%%%%%%%%%%%%%%%%%%%%%%%%%%%%%%%%%%%%%%%%%%%%%%%%%%%%%%%%%%%%%%%%%%%%%%%%%%%%%%
\nextgroupplot[title={{\smash{(c)}}},
scale only axis,
xmin=-110,
xmax=-70,
xlabel={$\rho\text{ [dBm]}$},
ymin=-19,
ymax=5,
y coord trafo/.code=\pgfmathparse{#1+20},
y coord inv trafo/.code=\pgfmathparse{#1-20},
ylabel={$\theta\text{ [dB]}$},
ytick={-20,-15,-10,-5,0,5},
grid=both]

%%%%%%%%%%%%%%%%%%%%%%%%%%%%%%%%%%%%%%%%%%%%%%%%%%%%%%%%
%%%% 			first value of kappa 				%%%%
\addplot[fill=blue1, draw=black] 
table[x = x, y=y]{simulation_results/figures/Pareto_regions/figc_shared_1.txt}\closedcycle;
\addplot [color=black,dotted, line width=1.5pt]
table[x = x, y=y]{simulation_results/figures/Pareto_regions/figc_shared_1.txt};

%\addplot[fill=blue1, draw=black] 
%table[x = x, y=y]{simulation_results/figures/Pareto_regions/figc_dedicated_1.txt}\closedcycle;
%\addplot [color=black,dotted,  line width=1.5pt]
%table[x = x, y=y]{simulation_results/figures/Pareto_regions/figc_dedicated_1.txt};

%%%%%%%%%%%%%%%%%%%%%%%%%%%%%%%%%%%%%%%%%%%%%%%%%%%%%%%%
%%%% 			Second value of kappa 				%%%%

\addplot[fill=green1, draw=black] 
table[x = x, y=y]{simulation_results/figures/Pareto_regions/figc_shared_4.txt}\closedcycle;
\addplot [color=black,dotted, line width=1.5pt]
table[x = x, y=y]{simulation_results/figures/Pareto_regions/figc_shared_4.txt};

%\addplot[fill=green1, draw=black] 
%table[x = x, y=y]{simulation_results/figures/Pareto_regions/figc_dedicated_4.txt}\closedcycle;
%\addplot [color=black,dotted,  line width=1.5pt]
%table[x = x, y=y]{simulation_results/figures/Pareto_regions/figc_dedicated_4.txt};

%%%%%%%%%%%%%%%%%%%%%%%%%%%%%%%%%%%%%%%%%%%%%%%%%%%%%%%%
%%%% 			third value of kappa 				%%%%

\addplot[fill=red1,  draw=black] 
table[x = x, y=y]{simulation_results/figures/Pareto_regions/figc_shared_8.txt}\closedcycle;
\addplot [color=black,dotted, line width=1.5pt]
table[x = x, y=y]{simulation_results/figures/Pareto_regions/figc_shared_8.txt};

%\addplot[fill=red1, draw=black] 
%table[x = x, y=y]{simulation_results/figures/Pareto_regions/figc_dedicated_8.txt}\closedcycle;
%\addplot [color=black,dotted,  line width=2pt]
%table[x = x, y=y]{simulation_results/figures/Pareto_regions/figc_dedicated_8.txt};

\node[below right, align=left]
at (rel axis cs:0.01,0.9) {\small $\kappa$ = (\textcolor{blue}{1}, \textcolor{green}{4}, \textcolor{red}{8})};
\draw[->, line width=0.5mm](rel axis cs:0.2,0.8) -- (rel axis cs:0.6,0.3);

\end{groupplot}

\end{tikzpicture}
	\fi
	\caption{Non-overflow frontiers of the shared (dedicated-EA) strategy represented by solid (dashed) lines.}
	\label{fig:Pareto_regions} 
\end{figure*}
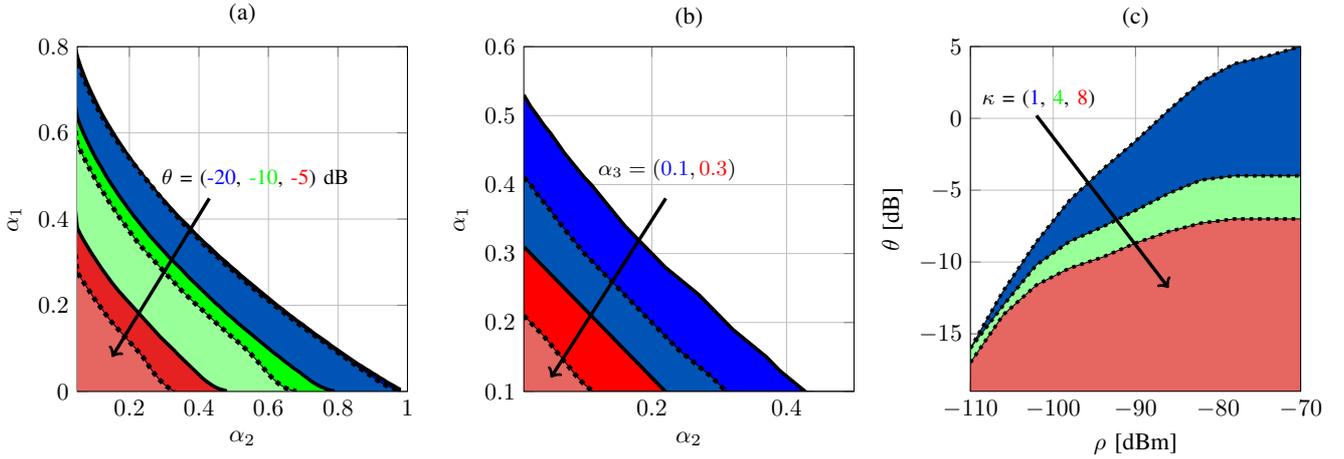

To showcase the network's stability regions, Fig. \ref{fig:Pareto_regions} presents the non-overflow region frontiers under shared and dedicated-EA strategies for different system parameters. Such regions ensure queues operating below the overflow threshold, which is represented via the filled area under the curves. The dark (solid lines) and light shaded (dashed lines) represent the shared and dedicated-EA allocation strategies, respectively. First, Fig. \ref{fig:Pareto_regions}(a) shows the relation between the arrival probability of the two highest priority classes ($\alpha_1$ and $\alpha_2$) and the \ac{SINR} threshold $\theta$. As explained in Fig. \ref{fig:TSP}, larger values of $\theta$ leads to higher aggregate network interference, thus, supporting lower traffic arrivals to operate within the non-overflow regions. We observe that for low values of $\theta$, the gap between the shared and dedicated-EA allocation diminishes, since the devices are able to empty their queues nearly easily even under strong mutual interference. As $\theta$ increases, the shared strategy outperforms the dedicated-EA, since more channels are available for each device for the former strategy. Similarly,  Fig. \ref{fig:Pareto_regions}(b) highlights the effect of increasing the third priority packets arrival probability, where the overflow region decreases with larger arrival probabilities. Such a figure can provide interesting insights when studying the relation between different classes of traffic in order to ensure a stable network. The performance comparison between the shared and the dedicated-EA strategies follows Fig. \ref{fig:Pareto_regions}(a). Finally, Fig. \ref{fig:Pareto_regions}(c) focuses on the relation between $\theta$, uplink power control threshold $\rho$ and  $\kappa$. For a given $\kappa$, we can expect that as the uplink transmission can operate under higher thresholds (i.e., higher probabilities), the feasible set of $\theta$ ensuring non-overflow operation increases till saturation is reached. This follows from the system transitioning from the noise limited to the interference limited scenario, which is governed by the value of $\sigma^2$. On the other hand, as $\kappa$ increases, the non-overflow region diminishes, which is due to the increased interference within the network. It is important to notice that the shared and dedicated-EA strategies provide similar ($\theta, \rho$) frontiers when considering the network's propagation parameters, since the main dynamics affecting this frontier is radio-related and is oblivious ot the adopted resource allocation strategy.

\thispagestyle{empty}
% !TEX root =../Integration.tex
%%%%%%%%%%%%%%%%%%%%%%%%%%%%%%%%%%%%%%%%%%%%%%%%%%%%%%%%%%%%%%%%%%%%%%%%%%%%%%%%%%
\section{Conclusion}\label{sec:Conclusion}

This paper presents a tractable and scalable spatiotemporal mathematical framework for large scale uplink prioritized multi-stream traffic in \ac{IoT} networks. The network is modeled via network of interacting vacation queue, where at the spatial macroscopic scale, interactions occur between different devices due to the mutual interference. At the spacial microscopic scale, interactions among different priority packets occur as the uplink channel can only be utilized by the highest priority packets at the device and is not available to any of the lower priority packets, which is denoted as service vacation. The developed spatiotemporal model is used to assess and compare three priority aware channel allocation strategies; namely dedicated-equal allocation, dedicated-weighted allocation and shared allocation strategy. Numerical evaluations showcase the performance of each priority class in terms of transmission success probability, average queue length, average-delay, delay distribution, and peak age of information. Furthermore, a multi-class priority agnostic scheme is used to benchmark the gains and costs of traffic prioritization on the different priority classes in terms of transmission success probability and average packet delay. The stability of the IoT network is assessed via the Pareto-frontiers of the non-overflow regions. Finally, results indicate the superiority of the shared channel allocation strategy over the dedicated ones, since the former offers higher pool of channels, enabling interference diversification. 

%, which shows the network parameters that ensures packet departure rates higher that the packet arrival rates for the different prioritized traffic streams. 

\thispagestyle{empty}
% !TEX root =../quant_cell_det_jnl.tex
%%%%%%%%%%%%%%%%%%%%%%%%%%%%%%%%%%%%%%%%%%%%%%%%%%%%%%%%%%%%%%%%%%%%%%%%%%%%%%%%%%
% Can use something like this to put references on a page
% by themselves when using endfloat and the captionsoff option.
%\ifCLASSOPTIONcaptionsoff
%  \newpage
%\fi
\bibliographystyle{./lib/IEEEtran.cls}
% argument is your BibTeX string definitions and bibliography database(s)
%\bibliography{IEEEabrv,../bib/paper}
\bibliography{./literature/Literature_Local}

% Generated by IEEEtran.bst, version: 1.14 (2015/08/26)
\begin{thebibliography}{10}
\providecommand{\url}[1]{#1}
\csname url@samestyle\endcsname
\providecommand{\newblock}{\relax}
\providecommand{\bibinfo}[2]{#2}
\providecommand{\BIBentrySTDinterwordspacing}{\spaceskip=0pt\relax}
\providecommand{\BIBentryALTinterwordstretchfactor}{4}
\providecommand{\BIBentryALTinterwordspacing}{\spaceskip=\fontdimen2\font plus
\BIBentryALTinterwordstretchfactor\fontdimen3\font minus
  \fontdimen4\font\relax}
\providecommand{\BIBforeignlanguage}[2]{{%
\expandafter\ifx\csname l@#1\endcsname\relax
\typeout{** WARNING: IEEEtran.bst: No hyphenation pattern has been}%
\typeout{** loaded for the language `#1'. Using the pattern for}%
\typeout{** the default language instead.}%
\else
\language=\csname l@#1\endcsname
\fi
#2}}
\providecommand{\BIBdecl}{\relax}
\BIBdecl

\bibitem{Palattella2016}
M.~R. {Palattella} \emph{et~al.}, ``Internet of things in the 5{G} era:
  Enablers, architecture, and business models,'' \emph{IEEE Journal on Selected
  Areas in Communications}, vol.~34, no.~3, pp. 510--527, March 2016.

\bibitem{3GPP2018IoT}
3GPP, ``{TR} 36.746 study on further enhancements to {LTE} device to device,
  user equipment to network relays for internet of things and wearables,''
  \emph{3rd Generation Partnership Project (3GPP), v15.1.1}, 2018.

\bibitem{Ayoub2018}
W.~{Ayoub} \emph{et~al.}, ``Internet of mobile things: Overview of {LoRaWAN,
  DASH7, and NB-IoT in LPWANs} standards and supported mobility,'' \emph{IEEE
  Communications Surveys Tutorials}, pp. 1--1, 2018.

\bibitem{elsawy2020spatial}
H.~Elsawy, M.~A. Kishk, and M.-S. Alouini, ``Spatial firewalls: Quarantining
  malware epidemics in large scale massive wireless networks,'' \emph{ArXiv},
  vol. abs/2006.05059, 2020.

\bibitem{3GPP2019}
3GPP, ``{TS} 23.203 policy and charging control architecture,'' \emph{3rd
  Generation Partnership Project (3GPP), v16.0.0}, 2019.

\bibitem{5GACIA2018}
5G-ACIA, ``{5G} for connected industries and automation,'' \emph{5G Alliance
  for Connected Industries and Automation}, 2018.

\bibitem{qbvieee2016}
IEEE, ``802.1{Qbv}-enhancements for scheduled traffic,'' \emph{Institute of
  Electrical and Electronics Engineers, Inc}, 2016.

\bibitem{Al-Fuqaha2015}
A.~{Al-Fuqaha} \emph{et~al.}, ``Internet of things: A survey on enabling
  technologies, protocols, and applications,'' \emph{IEEE Communications
  Surveys Tutorials}, vol.~17, no.~4, pp. 2347--2376, Fourth quarter 2015.

\bibitem{Bader2017}
A.~{Bader} \emph{et~al.}, ``First mile challenges for large-scale {IoT},''
  \emph{IEEE Communications Magazine}, vol.~55, no.~3, pp. 138--144, March
  2017.

\bibitem{S.Alfa2015}
A.~S.~Alfa, \emph{Applied discrete-time queues, second edition}.\hskip 1em plus
  0.5em minus 0.4em\relax Springer-New York USA, 2015.

\bibitem{Takagi1991}
\BIBentryALTinterwordspacing
H.~Takagi and Y.~Takahashi, ``Priority queues with batch poisson arrivals,''
  \emph{Operations Research Letters}, vol.~10, no.~4, pp. 225 -- 232, 1991.
  [Online]. Available:
  \url{http://www.sciencedirect.com/science/article/pii/016763779190063U}
\BIBentrySTDinterwordspacing

\bibitem{Doshi1986}
B.~T. Doshi, ``Queueing systems with vacations - a survey,'' \emph{Queueing
  Systems}, vol.~1, no.~1, pp. 29 -- 66, 1986.

\bibitem{Machihara1996}
F.~Machihara, ``A preemptive priority queue as a model with server vacations,''
  \emph{Journal of the Operations Research Society of Japan}, vol.~39, no.~1,
  pp. 118--131, 1996.

\bibitem{HarcholBalter2005}
M.~Harchol-Balter \emph{et~al.}, ``Multi-server queueing systems with multiple
  priority classes,'' \emph{Queueing Systems}, vol.~51, no.~3, pp. 331--360,
  Dec 2005.

\bibitem{Vuuren2007}
M.~V. {Vuuren} and I.~{Adan}, ``Approximate analysis of general priority
  queues,'' in \emph{Analysis of manufacturing systems}, 2007.

\bibitem{Sleptchenko2015}
A.~Sleptchenko \emph{et~al.}, ``Joint queue length distribution of multi-class,
  single-server queues with preemptive priorities,'' \emph{Queueing Systems},
  vol.~81, no.~4, pp. 379--395, Dec 2015.

\bibitem{Rao1988}
R.~R. {Rao} and A.~{Ephremides}, ``On the stability of interacting queues in a
  multiple-access system,'' \emph{IEEE Transactions on Information Theory},
  vol.~34, no.~5, pp. 918--930, Sep. 1988.

\bibitem{Luo1999}
{Wei Luo} and A.~{Ephremides}, ``Stability of n interacting queues in
  random-access systems,'' \emph{IEEE Transactions on Information Theory},
  vol.~45, no.~5, pp. 1579--1587, July 1999.

\bibitem{Andrews2011}
J.~G. {Andrews}, F.~{Baccelli}, and R.~K. {Ganti}, ``A tractable approach to
  coverage and rate in cellular networks,'' \emph{IEEE Transactions on
  Communications}, vol.~59, no.~11, pp. 3122--3134, November 2011.

\bibitem{ElSawy2013}
H.~{ElSawy}, E.~{Hossain}, and M.~{Haenggi}, ``Stochastic geometry for
  modeling, analysis, and design of multi-tier and cognitive cellular wireless
  networks: A survey,'' \emph{IEEE Communications Surveys Tutorials}, vol.~15,
  no.~3, pp. 996--1019, Third 2013.

\bibitem{Elsawy_tutorial}
H.~{ElSawy} \emph{et~al.}, ``Modeling and analysis of cellular networks using
  stochastic geometry: A tutorial,'' \emph{IEEE Communications Surveys
  Tutorials}, vol.~19, no.~1, pp. 167--203, Firstquarter 2017.

\bibitem{Zhong2017}
Y.~{Zhong}, T.~Q.~S. {Quek}, and X.~{Ge}, ``Heterogeneous cellular networks
  with spatio-temporal traffic: Delay analysis and scheduling,'' \emph{IEEE
  Journal on Selected Areas in Communications}, vol.~35, no.~6, pp. 1373--1386,
  June 2017.

\bibitem{Gharbieh2017}
M.~{Gharbieh} \emph{et~al.}, ``Spatiotemporal stochastic modeling of {IoT}
  enabled cellular networks: Scalability and stability analysis,'' \emph{IEEE
  Transactions on Communications}, vol.~65, no.~8, pp. 3585--3600, Aug 2017.

\bibitem{Gharbieh2018}
------, ``Spatiotemporal model for uplink {IoT} traffic: Scheduling and random
  access paradox,'' \emph{IEEE Transactions on Wireless Communications},
  vol.~17, no.~12, pp. 8357--8372, Dec 2018.

\bibitem{Yang2019}
H.~H. {Yang} and T.~Q.~S. {Quek}, ``Spatiotemporal analysis for {SINR} coverage
  in small cell networks,'' \emph{IEEE Transactions on Communications}, pp.
  1--1, 2019.

\bibitem{Chisci2019}
G.~{Chisci} \emph{et~al.}, ``Uncoordinated massive wireless networks:
  Spatiotemporal models and multiaccess strategies,'' \emph{IEEE/ACM
  Transactions on Networking}, pp. 1--14, 2019.

\bibitem{Chen2018}
Z.~{Chen} \emph{et~al.}, ``Throughput with delay constraints in a shared access
  network with priorities,'' \emph{IEEE Transactions on Wireless
  Communications}, vol.~17, no.~9, pp. 5885--5899, Sep. 2018.

\bibitem{Dester2019}
P.~S. {Dester} \emph{et~al.}, ``Performance analysis and optimization of a
  {$N$}-class bipolar network,'' \emph{IEEE Access}, vol.~7, pp.
  135\,118--135\,132, 2019.

\bibitem{ElSawy2014}
H.~{ElSawy} and E.~{Hossain}, ``On stochastic geometry modeling of cellular
  uplink transmission with truncated channel inversion power control,''
  \emph{IEEE Transactions on Wireless Communications}, vol.~13, no.~8, pp.
  4454--4469, Aug 2014.

\bibitem{Haenggi2012}
M.~Haenggi, \emph{Stochastic Geometry for Wireless Networks}.\hskip 1em plus
  0.5em minus 0.4em\relax New York, NY, USA: Cambridge University Press, 2012.

\bibitem{G.Kulkarni1999}
V.~G.~Kulkarni, ``Introduction to matrix analytic methods in stochastic
  modeling,'' \emph{Journal of Applied Mathematics and Stochastic Analysis},
  vol.~12, 01 1999.

\bibitem{Loynes1962}
R.~M. Loynes, ``The stability of a queue with non-independent inter-arrival and
  service times,'' \emph{Mathematical Proceedings of the Cambridge
  Philosophical Society}, vol.~58, no.~3, p. 497–520, 1962.

\bibitem{Krikidis2012}
I.~{Krikidis}, T.~{Charalambous}, and J.~S. {Thompson}, ``Buffer-aided relay
  selection for cooperative diversity systems without delay constraints,''
  \emph{IEEE Transactions on Wireless Communications}, vol.~11, no.~5, pp.
  1957--1967, May 2012.

\bibitem{Lee2014}
H.~Y. {Lee}, Y.~J. {Sang}, and K.~S. {Kim}, ``On the uplink {SIR} distributions
  in heterogeneous cellular networks,'' \emph{IEEE Communications Letters},
  vol.~18, no.~12, pp. 2145--2148, Dec 2014.

\bibitem{Haenggi_meta}
M.~{Haenggi}, ``The meta distribution of the {SIR} in {P}oisson bipolar and
  cellular networks,'' \emph{IEEE Transactions on Wireless Communications},
  vol.~15, no.~4, pp. 2577--2589, April 2016.

\bibitem{Elsawy_meta}
H.~{ElSawy} and M.~{Alouini}, ``On the meta distribution of coverage
  probability in uplink cellular networks,'' \emph{IEEE Communications
  Letters}, vol.~21, no.~7, pp. 1625--1628, July 2017.

\bibitem{Haenggi_meta2}
Y.~{Wang}, M.~{Haenggi}, and Z.~{Tan}, ``The meta distribution of the {SIR} for
  cellular networks with power control,'' \emph{IEEE Transactions on
  Communications}, vol.~66, no.~4, pp. 1745--1757, April 2018.

\bibitem{Zhou2016}
Y.~{Zhou} and W.~{Zhuang}, ``Performance analysis of cooperative communication
  in decentralized wireless networks with unsaturated traffic,'' \emph{IEEE
  Transactions on Wireless Communications}, vol.~15, no.~5, pp. 3518--3530, May
  2016.

\bibitem{Kaul2018}
S.~K. {Kaul} and R.~D. {Yates}, ``Age of information: Updates with priority,''
  in \emph{2018 IEEE International Symposium on Information Theory (ISIT)},
  June 2018, pp. 2644--2648.

\bibitem{Huang2015}
L.~{Huang} and E.~{Modiano}, ``Optimizing age-of-information in a multi-class
  queueing system,'' in \emph{2015 IEEE International Symposium on Information
  Theory (ISIT)}, June 2015, pp. 1681--1685.

\end{thebibliography}

\appendices

\section{Proof of Lemma 1}\label{se:Appendix_A}

In order to fully characterize the vacation period for a given priority class, the aggregate busy periods of higher priority queues need to be characterized. For the highest priority class (i.e., $i=1$), its transition matrix $P_1$ is that of a simple birth-death process. Consequently, its busy period, denoted as $\mathbf{V}_1$ is represented via the following absorbing Markov chain
\begin{equation}
\label{eq:P_1st}
\mathbf{V}_1=\left[ \begin{array}{llll} 
\bar{\alpha_1}\bar{p}_1 + \alpha_1 p & \alpha_1\bar{p}_1 &  &  \\
\bar{\alpha_1}p & \bar{\alpha_1}\bar{p}_1 + \alpha_1 p & \alpha_1\bar{p}_1 & \\
& & \ddots & \ddots  \\
& & \bar{\alpha_1}p & \bar{\alpha_1}\bar{p}_1 + \alpha_1
\end{array}\right].
\end{equation}
Let $\tilde{\mathbf{v}}_1 =  [\bar{\alpha_1}p_1 \; \mathbf{0}_{k_1}] \in \mathbb{R}^{k_1 -1 \times 1}$ denotes the absorption vector. Through $\mathbf{V}_1$ and $\tilde{\mathbf{v}}_1$, one can fully characterize the transitions when a first priority packet arrives as well as its successful departure (i.e., absorption). The second priority queue can be modeled as Geo/PH/1 queue, where the PH type distribution models the busy period of the first priority queue. Consider then the case of serving second priority packets, if a first priority packet arrives, an initialization vector $\mathbf{v}_1 = [ 1 \; \bm{0}_{k_1}]$ is required to characterize the states distribution and the probability of their occurrence $\chi_1$. Since the queue is initialized as empty, $\chi_1= \alpha_1$. The analysis for a generic $i$-th priority class is extended and with some mathematical adaptations, the lemma is finalized.

\iffalse
As a result, the second queue stochastic auxiliary matrix can be represented as
	\begin{equation}
	\tilde{\mathbf{S}}_2=\left[ \begin{array}{ll} 
	\bar{\bm{\chi}}_1 & \bm{\chi}_1\mathbf{v}_1 \\
	\tilde{\mathbf{v}}_1       & \mathbf{V}_1  
	\end{array}\right].
	\end{equation}
\fi
\section{Proof of Theorem 1}\label{se:Appendix_B}

For the dedicated access scheme, a packet belonging to the $i$-th priority queue will only experience aggregate interference from packets belonging to the same priority class.This packet will be granted transmission only if all the higher priority queues are empty. The portion of interfering device for the $i$-th queue at the \ac{BS} is $\mu\zeta_i$, where $\zeta_i = \sum_{z_i=1}^{k_{i}} \mathbb{P}\{(0,0,\cdots, 0,z_i)\}$ is the joint probability of having idle $i$-1 priority queues and non-idle $i$-th priority queue. Additionally, the adopted grant-free transmission scheme among the devices imposes a differentiation between the experienced interference into intra-cell and inter-cell interference, thus (\ref{eq:TSP_2}) is written as 
\begin{equation}\label{eq:laplace_0}
p = \text{exp}\Big \{ -\frac{\sigma^2\theta}{\rho} \Big\} \mathcal{L}_{I_{\text{out}}}\Big (\frac{\theta}{\rho} \Big) \mathcal{L}_{I_{\text{in}}}\Big (\frac{\theta}{\rho}\Big).
\end{equation} 

Since full channel inversion power control with threshold $\rho$ is employed, two main results hold: (i) received power from the devices at a given \ac{BS} equals $\rho$ (ii) interference power from the neighboring devices is strictly lower than $\rho$. Following \cite[Theorem 1]{ElSawy2014}, the LT of the aggregate inter-cell interference at the serving \ac{BS} for an $i$-th priority packet is 
\begin{equation}\label{eq:laplace_1}
\mathcal{L}_{I_{\text {out}}}(s) \approx \exp \Big(-2 \pi\mu\zeta_i s^{\frac{2}{\eta}} \mathbb{E}_{P}\left[P^{\frac{2}{\eta}}\right] \int_{(s \rho)^{\frac{-1}{\eta}}}^{\infty} \frac{y}{y^{\eta}+1} d y\Big),
\end{equation}
where the approximation is due to the assumed independent transmission powers of the devices (Approximation~\ref{approx1}(i)). The LT of the inter-cell interference can be evaluated as \cite[Lemma 1]{Gharbieh2017}
\begin{equation}\label{eq:laplace_2}
\mathcal{L}_{I_{\text {in}}}(s) \approx \mathbb{P}\{\mathcal{N}=0\}+\sum_{n=1}^{\infty} \frac{\mu^{n}(\lambda c)^{c}\mathrm{\Gamma}(n+c)}{(1+s \rho)^{n}\mu+\lambda c)^{n+c}\mathrm{\Gamma}(n+1) \mathrm{\Gamma}(c)},
\end{equation}
where $\mathrm{\Gamma}(\cdot)$ is the gamma function, $\mathcal{N}$ is a random variable representing the number of neighbors and $c=3.575$ is a constant defined to approximate Voronoi cell's PDF in $\mathbb{R}^2$. Plugging (\ref{eq:laplace_1}) and (\ref{eq:laplace_2}) into (\ref{eq:laplace_0}) and following \cite[Lemma 1]{Gharbieh2017}, the theorem is derived.

\end{document}